\documentclass[pra,showpacs,onecolumn,notitlepage,superscriptaddress,nofootinbib]{revtex4-1}
\usepackage[utf8]{inputenc}
\usepackage[english]{babel}
\usepackage[T1]{fontenc}

\usepackage{graphicx}
\usepackage{amssymb,amsfonts,amsmath,amsthm}
\usepackage{braket}
\usepackage{hyperref}
\usepackage{color}
\usepackage{soul}

\usepackage{enumerate,enumitem}
\usepackage{comment}
\usepackage{hyperref}
\usepackage{eqnarray}

\usepackage{tikz}
\usetikzlibrary{shapes.geometric, arrows, shapes.multipart, positioning}
\tikzset{mynode/.style={draw,circle,scale=0.7,inner sep=2pt,outer sep=0pt}}

\newcommand{\subscript}[2]{$#1 _ #2$}

\newcommand{\io}[0]{\underset{\mathrm{IO}}{\rightarrow}}

\newcommand{\lrio}[0]{\underset{\mathrm{IO}}{\leftrightarrow}}

\newtheorem*{rep@theorem}{\rep@title}
\newcommand{\newreptheorem}[2]{%
	\newenvironment{rep#1}[1]{%
		\def\rep@title{#2 \ref{##1}}%
		\begin{rep@theorem}}%
		{\end{rep@theorem}}}
\makeatother

\newtheorem{proposition}{Proposition}
\newtheorem{lemma}[proposition]{Lemma}
\newtheorem{definition}[proposition]{Definition}

\newtheorem{corollary}[proposition]{Corollary}

\newreptheorem{proposition}{Proposition}
\newreptheorem{corollary}{Corollary}

\def\B{\mathcal{B}}
\def\I{\mathcal{I}}
\def\S{\mathcal{S}}

\def\H{\mathcal{H}}
\def\P{\mathcal{P}}

\def\O{\mathcal{O}}
\def\D{\mathcal{D}}
\def\U{\mathcal{U}}
\def\F{\mathcal{F}}

\def\C{\mathbb{C}}
\def\R{\mathbb{R}}

\def\Tr{\operatorname{Tr}}
\def\mcs{\mathrm{mcs}}

\def\eg{\textit{e.g.} }
\def\mudown{\mu^\downarrow}

\def\croof{\mathrm{cr}}
\def\op{\mathrm{top}}
\def\cv{\mathrm{cv}}

\def\Ud{\mathcal{U}^{\mathrm{psd}}}
\def\Uo{\mathcal{U}^{\mathrm{psc}}}

\def\vc{\nu}

\newcommand\bigzero{\makebox(0,0){\text{\huge0}}}

\newcommand{\rvline}{\hspace*{-\arraycolsep}\vline\hspace*{-\arraycolsep}}
\DeclareMathOperator*{\argmin}{arg\,min}
\usepackage{xcolor}

\begin{document}

\title{Generalized coherence vector applied to coherence transformations and quantifiers}

\author{G. M. Bosyk}
\email{Corresponding author: gbosyk@gmail.com }
\affiliation{Universit\`{a} degli Studi di Cagliari, Cagliari, Italy}
\affiliation{Instituto de F\'isica La Plata, UNLP, CONICET, La Plata, Argentina}
\affiliation{Grupo de Matemática Aplicada, Departamento de Matemática, Universidad CAECE, CABA, Argentina}
\author{M. Losada}
\affiliation{Universit\`{a} degli Studi di Cagliari, Cagliari, Italy}
\affiliation{Universidad Nacional de C\'{o}rdoba, FAMAF, CONICET,  C\'ordoba, Argentina}
\author{C. Massri}
\affiliation{Instituto de Investigaciones Matemáticas ``Luis A. Santal\'o'', UBA, CONICET, CABA, Argentina}
\affiliation{Grupo de Matemática Aplicada, Departamento de Matemática, Universidad CAECE, CABA, Argentina}
\author{H. Freytes}
\affiliation{Universit\`{a} degli Studi di Cagliari, Cagliari, Italy}
\author{G. Sergioli}
\affiliation{Universit\`{a} degli Studi di Cagliari, Cagliari, Italy}
%


\begin{abstract}
	One of the main problems in any quantum resource theory is the characterization of the 
	conversions between resources by means of the free operations of the theory.
	In this work, we advance on this characterization within the quantum coherence resource theory by introducing the generalized coherence vector of an arbitrary quantum  state.
	The generalized coherence vector is a probability vector that can be interpreted as a concave roof extension of the pure states coherence vector.
	We show that it completely characterizes the notions of being incoherent, as well as being maximally coherent.
	Moreover, using this notion and the majorization relation, we obtain a necessary condition for the conversion of general quantum states by means of incoherent operations.
	These results generalize the necessary conditions of conversions for pure states given in the literature, and show that the tools of the majorization lattice are useful also in the general case.
	Finally, we introduce a family of coherence quantifiers by considering concave and symmetric functions applied to the generalized coherence vector. We compare this proposal with the convex roof measure of coherence and others quantifiers given in the literature.
\end{abstract}


\date{\today}

\maketitle

\section{Introduction}

Quantum coherence is one of the fundamental aspects of the quantum theory. It has practical relevance in numerous fields of quantum physics, particularly in quantum information processing.
Moreover, quantum coherence is considered as a quantum resource that can be converted, manipulated and quantified \cite{Aberg2006,Baumgratz2014}. 
It admits a resource-theoretic formulation in terms of incoherent states (free states), coherent states (resources) and incoherent operations (free operations).

Since coherence is a basis dependent concept, the three components of the resource-theoretic formulation have to be defined for a given incoherent basis.
In the standard approach, the incoherent basis is an orthonormal basis (see \eg \cite{Baumgratz2014,Streltsov2017}). 
There are also alternative resource-theoretic formulations based on non-orthonormal basis or positive-operator-valued measures (see \eg \cite{Theurer2017,Rastegin2018,Bischof2019}).

In this work, we follow the standard formulation. 
The incoherent states are diagonal in the incoherent basis, whereas coherent states have off-diagonal elements in this basis. 
Regarding the free operations, there is not a unique definition.
Several definitions, often motivated by their operational interpretations, have been introduced (see, e.g., \cite{Streltsov2017} for a review). 
In what follows, we restrict our attention to the definition of incoherent operation introduced in \cite{Baumgratz2014}.
Within this definition, quantum coherence cannot be created from any incoherent input state by means of incoherent operations, not even in a probabilistic way.

One of the main problems in any resource theory is characterizing the conversion between states by means of free operations \cite{Chitambar2019}.
In the quantum coherence case, this problem has been completely solved for 
incoherent transformations from pure to pure states (see Refs.~\cite{Du2015,Du2015b,Du2017,Chitambar2016,Zhu2017} or Lemma~\ref{lemma:pure_states_io}),
as well as for transformations from pure to mixed states (see  Refs.~\cite{Du2019,Zhu2017} or Prop.~\ref{lemma:theorem4du2019}). 
This characterization is given in terms of the majorization relation \cite{MarshallBook} between the coherence vectors of the pure states. 
Motivated by this fact, we propose a generalization of the coherence vector applicable to arbitrary quantum states, and we advance on the characterization of the state conversion by means of incoherent operations by appealing to the majorization lattice theory \cite{Bapat1991,Cicalese2002,Bosyk2019,Massri2020}.
More precisely, given a pure state decomposition of a quantum state, we define the coherence vector of the decomposition in terms of the coherence vectors of the pure states. Then, we define the coherence vector
of a general quantum state as the supremum (in terms of the majorization order relation) of the coherence vectors of all pure-state decompositions.
In this way, our proposal can be interpreted as a concave roof  extension of the pure state case.
Alternatively, the generalized coherence vector of an arbitrary state $\rho$ can be also defined as the supremum of all coherence vectors of the pure states that can be converted into  $\rho$ by means of an incoherent operation. 
Hence, our definition also acquires an operational meaning. 

We prove that the generalized coherence vector characterizes the notions of being incoherent, as well as being maximally coherent. 
In addition, we extend the necessary condition of Prop.~\ref{lemma:theorem4du2019} (see Refs.~\cite{Du2019,Zhu2017}) 
to the case of initial mixed states, which is also given in terms of the majorization relation of the corresponding coherence vectors. 
This result is a step forward on the characterization of conversions between general quantum states under incoherent operations, whose complete solution is only known for the single qubit system \cite{Streltsov2017b,Shi2017}. Indeed, in higher dimensions ($d \geq 4$), it was recently shown that a finite number of conditions in terms of coherence measures are not sufficient to fully characterize coherence transformations between general quantum states \cite{Du2019}. 
Thus, the complete characterization of the general case remains open.

Another main problem in any resource theory is to quantify the resource amount of any state \cite{Chitambar2016}. 
There are several coherence quantifiers and each of them captures different operational aspects of coherence, for instance, the distillable coherence, the coherence cost \cite{Yuan2015,Winter2016}, the relative entropy of coherence and the $\ell_1$-norm of coherence \cite{Baumgratz2014}, among others (see \eg \cite{Streltsov2017}). 
Providing new quantifiers of coherence is an ongoing topic in the resource theory of coherence.
A common strategy for obtaining a coherence quantifier is to define a suitable function on the pure states and then extend it to the entire set of quantum states.
The extension can be done in different ways. The most frequently used is the convex roof construction\cite{Du2015b,Zhu2017}, which was originally applied in the entanglement theory \cite{Bennett1996,Vidal2000}. 
A recent proposal was given in \cite{Yu2020}, based on 
the state conversion process from pure to arbitrary quantum states by means of incoherent operations.
In this work, we also present a different approach to obtain a family of coherence quantifiers, based on the generalized coherence vector.

This paper is organized as follows. 
In Sec. \ref{sec:review_coherence}, we recall the basics elements of the resource theory of quantum coherence. 
In particular, we review the notions of incoherent and coherent states, and incoherent operations. 
In addition, we present some important results about conversions of coherent sates, as well as its axiomatic quantification, focusing on coherence measures based on the convex roof construction and on coherence monotones recently introduced. 
In Sec. \ref{sec:coherence_vector}, we introduce the notion of generalized coherence vector, valid for arbitrary quantum states. We show that it is a good definition, since it allows to characterize the notions of being incoherent and maximally coherent. 
In Sec. \ref{sec:necessary conditions}, we apply the generalized coherence vector to provide a necessary condition, 
in terms of a majorization relation, for the conversion of general quantum states. 
In Sec. \ref{sec:familiy_monotones}, we introduce a family of monotones based on the coherence vector, and we compare it with the convex roof construction and other monotones introduced in the literature. 
In Sec. \ref{sec:examples}, we applied this family of monotones to quantify the coherence of a qubit system and a maximally coherent qutrit going through a depolarizing channel. 
Finally, some concluding remarks are given in Sec \ref{sec:conclusions}. 
For the sake of readability, auxiliaries lemmas and proofs are presented separately in the appendices \ref{sec: App1} and \ref{sec: App2}, respectively.

\section{Preliminaries: Resource theory of quantum coherence}
\label{sec:review_coherence}

In this section, we review the resource theory of quantum coherence introduced in \cite{Baumgratz2014}. 
In what follows, we consider a quantum system represented by a $d$-dimensional Hilbert space $\H$. Moreover, we denote as $\S(\H)$ the set of  density operators and as $\P(\H)$ the set of pure states.
Since the coherence of a quantum state is a basis dependent notion, it is necessary to choose a reference basis in order to formulate its resource theory, which is usually called \textit{incoherent basis}.
In the rest of this work, we will choose the computational basis $\B = \{\ket{i} \}^{d-1}_{i =0}$ as the incoherent basis.

\subsection{Free states, resources and free operations}

Any resource theory is built from the basic notions of free states, resources and free operations.  
In the case of the resource theory of coherence, the free states are  quantum states with diagonal density matrix in the incoherent basis, i.e., a state $\rho$ is incoherent if and only if $\rho = \sum_{i=0}^{d-1} p_i \ket{i}\bra{i}$, with  $\sum_{i=0}^{d-1} p_i =1$ and $p_i\geq 0$ for all $i \in \{0, \ldots, d-1 \}$. We call them \textit{incoherent states}, and we denote the set of incoherent states as $\I$. 
The resources of a theory are the states which are not free. In the coherence case, the resources are quantum states represented by non-diagonal density matrices in the incoherent basis. We call them 
\textit{coherent states}.
Regarding the free operations, several definitions have been introduced \cite{Streltsov2017}. 
For each definition, we obtain different resource theories for coherence. 
In what follows, we focus on the incoherent operations introduced in \cite{Baumgratz2014}. 
  
In order to define the free operations, we consider completely positive and  trace-preserving maps (CPTP) defined on the set of density operators 
$\S(\H)$. 
If $\Lambda: \S(\H) \mapsto \S(\H)$ is a CPTP map, it has an operator-sum representation in terms of Kraus operators $\{ K_n \}_{n= 1}^N$  of the form $\Lambda(\rho) =  \sum_{n= 1}^{N} K_n \rho K^\dag_n$, where Kraus operators are such that $\sum_{n= 1}^{N}K_n^{\dag}K_n = I$ (with $I$ the identity of the Hilbert space).
The free operations for any resource theory of coherence have to be CPTP maps satisfying, at least, the additional condition of not creating coherence from an incoherent state.
More precisely, $\Lambda(\rho) \in \I$ for any $\rho \in \I$. 
All operations of this type form the set of maximally incoherent operations (MIO).

In this work, we are interested in a subset of the maximally incoherent operations, the so-called \textit{incoherent operations} (IO), which were introduced in \cite{Baumgratz2014}. IO can be defined in terms of Kraus operatiors as follows \cite{Yadin2016,Winter2016,Yao2017}:	
	
\begin{definition}[\textbf{Incoherent operation}]
	\label{def:IO}
	A CPTP map $\Lambda$ is an incoherent operation if it admits a Kraus representation $\{ K_n \}_{n= 1}^N$, such that the Kraus operators  are incoherent, that is, 
	$K_n \ket{i} \propto  \ket{f_n(i)}$, for all $n \in \{1, \ldots, N \}$, with $f_n$ a relabeling function of the set $\{0,\ldots, d-1\}$.
\end{definition}

\subsection{Necessary and sufficient conditions for coherent transformations}

In this subsection, we recall some important results about state transformations under incoherent operations. 
We denote as $\rho \io \sigma$ whenever a state $\rho$ can be transformed into an state $\sigma$ by means of an incoherent operation, i.e., when there is an incoherent operation $\Lambda$ such that $\sigma= \Lambda(\rho)$.

We note that any incoherent state can be reached by any other state by means of a suitable incoherent operation, that is, for any state $\sigma \in \I$ there exists a state $\rho$ such that $\rho \io \sigma$. 
Moreover, there are some states that can be converted into any other state (not necessarily incoherent) by means of incoherent operations. More precisely, there exist states $\rho$ called \textit{maximally coherent state} (MCS), such that $\rho \io \sigma$ for any $\sigma \in \S(\H)$.
The canonical MCS state is a pure state of the form $\rho^{\mcs} = \ket{\Psi^{\mcs}}\bra{\Psi^{\mcs}}$ with $\ket{\Psi^{\mcs}} = \sum_{i=0}^{d-1} \frac{1}{\sqrt{d}} \ket{i}$.
The set of all MCSs can be obtained from the orbit of $\rho^{\mcs}$ under the set of unitary incoherent operations, which are given by operators of the form $U_{\mathrm{IO}} = \sum_{i= 0}^{d-1} e^{\imath \theta_{i}} \ket{\pi(i)}\bra{i}$, 
where $\pi$ is a permutation acting on the set $\{0,\ldots,d-1\}$ and $\theta_{i} \in \R$\cite{Peng2016}.

In order to address the general problem of state transformation, we need the following notions. Let $\Delta_{d}$ be the set of $d$-dimensional probability vectors, i.e., 
\begin{equation}
\Delta_{d}= \big\{(u_0, \ldots, u_{d-1}) \in \mathbb{R}^d: u_i \geq 0,  \sum_{i= 0}^{d-1} u_i =1 \big\},
\end{equation}
and let $\Delta^\downarrow_{d} \subseteq \Delta_{d}$ be the set of $d$-dimensional probability vectors with their entries decreasingly ordered.
The \textit{coherence vector of a pure state} of a $d$-dimensional Hilbert space is a probability vector in $\Delta_{d}$ defined as follows:
\begin{definition}[\textbf{Coherence vector}]
	\label{def:coherence_vector_pure}
	Let $\B = \{\ket{i} \}^{d-1}_{i =0}$ be the incoherent basis.
	The coherence vector of a pure state $\ket{\psi}\bra{\psi}$ is defined as
	\begin{equation}
	\label{eq:cohrence_vector_pure}
	\mu(\ket{\psi}\bra{\psi}) =  \left(|\braket{0|\psi}|^2, \ldots, |\braket{d-1|\psi}|^2\right).
	\end{equation}
\end{definition}
We also define the ordered coherence vector $\mudown(\ket{\psi}\bra{\psi}) \in \Delta^\downarrow_{d}$, which is given by the entries of the vector $\mu(\ket{\psi}\bra{\psi})$, but in a non-increasing order.

The state transformations between quantum  states is related with the \textit{majorization relation of probability vectors}. 
The majorization relation is defined as follows (see \eg \cite{MarshallBook}).

\begin{definition}[\textbf{Majorization relation}]
	 Given $u,v \in \Delta_{d}$, it is said that $u$ is majorized by $v$ (denoted as $u \preceq v$) if, and only if, $\sum_{i=0}^{k} u_{\pi_u(i)} \leq \sum_{i=0}^{k} v_{\pi_v(i)}$, for all $k \in \{0, \ldots, d-1\}$, where $\pi_u$ and $\pi_v$  are permutations acting on the set $\{0, \ldots, d-1\}$ that sort the entries of $u$ and $v$, respectively, in a non-increasing order.
\end{definition}
The majorization relation is a preorder on the set $\Delta_{d}$ and a partial order on the set $\Delta^\downarrow_{d}$. Moreover, the set $\Delta^\downarrow_{d}$ endowed with the majorization relation $\preceq$ is a complete lattice\footnote{A preorder relation is a reflexive and transitive binary relation, and a partial order relation is a preorder that it is also antisymmetric.	A set $P$ endowed with a partial order relation is a complete lattice if the supremum and infimum of any subset of $P$ exist (see \eg \cite{DaveyBook}).}  \cite{Bapat1991,Bosyk2019}, and it is called the \textit{majorization lattice}. 

The algorithms to obtain the supremum and infimum of any subset of the majorization lattice can be found in \cite{Cicalese2002,Bosyk2019,Massri2020}.
In particular, the supremum of a set $\U\subseteq\Delta_d^\downarrow$, denoted as $\bigvee \U$, can be computed as follows. 
First, we obtain the \textit{Lorenz curve}\footnote{The Lorenz curve of a probability vector $u \in \Delta_d$ is an increasing and concave function $L_u: [0,d] \to [0,1]$ formed by the linear interpolation of the points $\{(j,s_{j}(u^\downarrow))\}_{j= 0}^d$. It can be shown that $u \preceq v \iff L_u \leq L_v$ (see \eg \cite{MarshallBook}).} of $\bigvee \U$, denoted as $L_{\bigvee \U}$. In \cite{Bosyk2019} it has been shown that $L_{\bigvee \U}$ is equal to the the upper envelope\footnote{We recall that the upper envelope of a continuous function $f : \mathbb{R} \to \mathbb{R}$ is defined as $\inf\{g: f \leq g \ \text{and} \ g \ \text{is continuos and concave} \}$ (see \eg \cite[Def.4.1.5]{BratelliBook}).} of the polygonal curve given by the linear interpolation of the set of points $\{(j,S_{j})\}_{j= 0}^d$, where $S_j = \sup\{s_j(u) : u \in \U \}$ and $s_j(u) =\sum^{j-1}_{i=0} u_i$, with the convention $S_{0} = 0$.
Finally, we have $\bigvee \U = (L_{\bigvee \U}(1),L_{\bigvee \U}(2)-L_{\bigvee \U}(1),\ldots,L_{\bigvee \U}(d)-L_{\bigvee \U}(d-1))$.

We remark that $\bigvee \U$ may or may not belong to $\U$. 
When $\bigvee \U \in \U$, $\bigvee \U$ is a maximum.
In this case, we have $S_k = L_{\bigvee \U}(k) = s_k(\bigvee \U)$ for all $k \in \{1, \ldots, d-1\}$. In other words, the Lorenz curve of $\bigvee \U$ is just the linear interpolation of $\{(j,S_{j})\}_{j= 0}^d$.

The majorization relation is intimately related with \textit{Schur-concave functions} (see \eg \cite[I.3]{MarshallBook}), which are functions that anti-preserves the preorder relation. More precisely, a function $f: \Delta_d \to \mathbb{R}$ is Schur-concave if, for all $u, v \in \Delta_d$ such that $u \preceq v$, $f(u) \geq f(v)$.  
Moreover, if the function $f$ also satisfies that $f(u) > f(v)$ whenever $u$ is strictly majorized by $v$ (i.e., when $u \preceq v$ and $u \neq \Pi v$, with $\Pi$ a permutation matrix), we say that it is strictly Schur-concave.
In particular, the generalized entropies, including Shannon, R\'eny and Tsallis entropies, are examples of strictly Schur-convave functions (see \eg \cite{Bosyk2016}). 

Taking into account these definitions, we  present the following results about necessary and sufficient conditions for coherent transformations.   
The first result completely characterizes the incoherent transformations between pure states in terms of the majorization relation between their corresponding coherence vectors (see \cite{Du2015,Du2015b,Du2017,Zhu2017,Chitambar2016}).
\begin{proposition}
	\label{lemma:pure_states_io}	
	Let $\ket{\psi}\bra{\psi}$ and $\ket{\phi}\bra{\phi}$ be two arbitrary pure states, and let $\Lambda$ be an incoherent operation.
	Then,
	\begin{equation}
	\label{eq:pure_states_io}
	\ket{\psi}\bra{\psi} \io \ket{\phi}\bra{\phi} \iff \mu(\ket{\psi}\bra{\psi}) \preceq \mu(\ket{\phi}\bra{\phi}).
	\end{equation}
\end{proposition}
Notice that if both transformations are possible, we have $\ket{\psi}\bra{\psi} \lrio \ket{\phi}\bra{\phi} \iff \mu(\ket{\psi}\bra{\psi}) = \Pi \left(\mu(\ket{\phi}\bra{\phi})\right)$, with $\Pi$ a permutation matrix. 
As a consequence, the coherence vector $\mu(\ket{\psi}\bra{\psi})$ of the pure state $\ket{\psi}\bra{\psi}$ and its ordered probability vector $\mudown(\ket{\psi}\bra{\psi})$ are equivalent for the coherence resource theory.

The next result, given in \cite[Th. 4]{Du2019}, is a generalization of the previous proposition. It provides necessary and sufficient conditions for transformations from pure states to arbitrary states by means of incoherent operations.

\begin{proposition}
	\label{lemma:theorem4du2019}	
	Let $\ket{\psi}\bra{\psi}$ be an arbitrary pure state and $\sigma$ be an arbitrary quantum state.
	Then,
	\begin{equation}
	\label{eq:theorem4du2019}
	\begin{split}
	\ket{\psi}\bra{\psi} \io \sigma &\iff \exists \{  p_n, \ket{\phi_{n}}\}_{n= 1}^N \ \ \text{such that} \\
	 &\begin{array}{l}
		(1) \ \sigma= \sum_{n = 1}^{N} p_n \ket{\phi_{n}}\bra{\phi_{n}} \ \text{and} \\
		(2) \ \mu(\ket{\psi}\bra{\psi})  \preceq \sum_{n = 1}^{N} p_{n} \mu^\downarrow(\ket{\phi_{n}}\bra{\phi_{n}}). 
	\end{array}
	\end{split}
	\end{equation}
\end{proposition}

A related result, given in \cite[Lemma 4]{Zhu2017}, provides a particular decomposition of the final state $\sigma$ given in  Prop.~\ref{lemma:theorem4du2019},
\begin{equation}
\label{eq:lemma4zhu}
\ket{\psi}\bra{\psi} \io \sigma \implies \mu(\ket{\psi}\bra{\psi}) \preceq  \sum_{n = 1}^{N} p_{n} \mu^\downarrow(\ket{\phi_{n}}\bra{\phi_{n}}), 
\end{equation}
where $p_{n} = \Tr(K_n \ket{\psi}\bra{\psi} K_n^\dag)$,  $\ket{\phi_{n}}\bra{\phi_{n}} = K_n \ket{\psi}\bra{\psi}K^\dag_n/p_{n}$,  and 
$\{K_n\}_{n= 1}^N$ are the incoherent Kraus operators of the incoherent operation $\Lambda$, which satisfies $\sigma= \Lambda(\ket{\psi}\bra{\psi})$.

The result given in Prop. \ref{lemma:pure_states_io} is a particular case of Prop. \ref{lemma:theorem4du2019}, but in the former the incoherent transformations are fully characterize by the majorization relation between the corresponding coherence vectors of the pure states.

\subsection{Coherence measures}

In this subsection, we introduce the notion of \textit{coherence measures}, mainly based on the axiomatic formulation given in \cite{Baumgratz2014}.

\begin{definition}[\textbf{Coherence measure}]
\label{def:coherence_measures}
A coherence measure is a function $C : \S(\H) \to \R$ satisfying the following conditions:
\begin{enumerate}[label=(\subscript{\rm{C}}{{\arabic*}})]
	\item Vanishing on incoherent states: 
	$C\left(\rho \right)  = 0$ for any $\rho$ incoherent. \label{c1:nonneg}
	\item Monotonicity under incoherent operations: $C(\rho) \geq C(\Lambda(\rho))$ for any incoherent operation $\Lambda$ and any state $\rho$.  \label{c2:monotonicity}
	\item Monotonicity under selective incoherent operation: $C(\rho) \geq \sum_{n = 1}^{N}  p_n C(\sigma_n)$, for any state $\rho$ and any incoherent operation $\Lambda$, with incoherent Kraus operators $\{K_n\}_{n= 1}^N$, where $p_n = \Tr{K_n \rho K^\dag_n}$ and $\sigma_n = K_n \rho K^\dag_n / p_n$.  \label{c3:strong_monotonicity}
	\item Maximal coherence: $\arg\max_{\rho \in S(\H)} C(\rho)$ is reached at maximally coherent states.  \label{c5:normalization}
	\item Convexity: $C\left( \sum_{k = 1}^{M}  q_k \rho_k\right)  \leq \sum_{k = 1}^{M} q_k C(\rho_k)$ .  \label{c4:convexity}
\end{enumerate}
\end{definition}
Condition \ref{c1:nonneg} guarantees that the measure is well defined for the incoherent states.
Condition \ref{c2:monotonicity} ensures that it is consistent with incoherent operations. Both are the minimal requirements for any quantifier of the coherence resource.
Condition \ref{c3:strong_monotonicity} guarantees that coherence does not increase under incoherent measurements, 
even if one has access to the individual measurement
outcomes.
When a quantifier satisfies the conditions  \ref{c1:nonneg}--\ref{c3:strong_monotonicity}, it is called \textit{coherence monotone}.
We have included the condition \ref{c5:normalization} because maximally coherent states are the golden unit of the coherence resource theory with the incoherent operations given in Def. \ref{def:IO} (the golden unit does not necessary exist for other set of free operations, see \eg \cite{Streltsov2017}). The relevance of this condition is discussed in \cite{Peng2016}. 
Finally, condition \ref{c4:convexity} is often related with the fact that mixing states does not increase the amount of coherence.  
Although the convexity condition \ref{c4:convexity} is a desirable property for coherence quantifiers, it is not considered essential.
Indeed, there are important quantifiers of coherence that do not satisfy \ref{c4:convexity}, such as the maximum relative entropy of coherence  \cite{Bu2017}. 
Finally, we note that when conditions \ref{c2:monotonicity} and \ref{c4:convexity} are satisfied, condition \ref{c3:strong_monotonicity} is automatically satisfied.

There are several quantifiers of coherence that satisfy some or all of the conditions given in Def.~\ref{def:coherence_measures}.
In this work, we are interested in families of coherence measures constructed from quantifiers of pure states. Before
introducing an important result for coherence 
measures restricted to pure states (see \eg \cite{Du2015b,Zhu2017}), we need to define the following set of functions:
\begin{multline}
\label{eq:setofF}
\F= \big\{f : \mathbb{R}^d \to [0,1]: f \ \text{is symmetric and concave,}  f(1, 0, \ldots, 0) =0 \ \text{and} \ \arg\max_{u \in \mathbb{R}^d} f(u)= (1/d,  \ldots, 1/d) \big\}.
\end{multline}
Since a symmetric and concave function is Schur-concave \cite{MarshallBook}, then all functions in $\F$ are Schur-concave.

The following result guarantees that the restriction of any  coherence monotone to pure states can be written in terms of a function belonging to $\F$ evaluated on the coherence vectors of the pure states (see \eg \cite{Du2015b,Zhu2017}). 
We will call as $f_C$ to the associated function of the coherence monotone $C$.
\begin{proposition}
	\label{lemma:cohrence_measures_pure}	
	Given a coherence monotone $C : \S(\H)  \to \R$  satisfying conditions \ref{c1:nonneg}--\ref{c5:normalization}, there exists a function $f_{C} \in \F$, such that the restriction of $C$ to the pure states, denoted as $C|_{\P(\H)}$, can be written as
	\begin{equation}
	\label{eq:f_C}
	C|_{\P(\H)}(\ket{\psi}\bra{\psi}) = f_{C}(\mu(\ket{\psi}\bra{\psi}).
	\end{equation}
\end{proposition}

Conversely, given a function $f \in \F$, it is possible to define a coherence monotone. 
In the literature there are at least two proposals to do this. One was introduced in \cite{Du2015b,Zhu2017}, whereas the other was recently developed in \cite{Yu2020}.

The first proposal appeals to the convex roof construction (see \eg \cite{Uhlmann2010}). 
Before introducing the \textit{convex roof measure of coherence}, we define the set of all pure state decompositions of a given state $\rho$, 
\begin{equation}
\label{eq:set_pure_ensambles}
\D(\rho)= \left\{ {\left\lbrace  q_k, \ket{ \psi_k} \right\rbrace }_{k = 1}^M : \ \rho= \sum_{k = 1}^{M}  q_k \ket{\psi_k}\bra{\psi_k}\ \right\},	
\end{equation}
where $\left(q_1, \dots, q_M \right) \in \Delta_M$  and $\ket{\psi_k} \in \H$ are unit-normed vectors (but not necessarily orthogonal to each other).
	
A complete characterization of this set is given by the \textit{Schr\"odinger mixture theorem} (also known as \textit{classification theorem for ensembles}, see \eg \cite{Hughston1993,Nielsen2000}). 
More precisely, $\{p_k, \ket{\psi_k}\}_{k = 1}^M \in \D(\rho)$ if, and only if, there exists a unitary matrix $U$ of $M \times M$ ($M \geq d$) such that
\begin{equation}
	\label{eq:SchTheo}
	\ket{\psi_k} = \frac{1}{\sqrt{q_k}} \sum_{j=1}^{d}  \sqrt{\lambda_j} U_{k,j} \ket{e_j},
\end{equation}
where $\lambda_j$ and $\ket{e_j}$ are the eigenvalues and eigenstates of $\rho$, respectively.

Now, we introduce the convex roof measure of coherence (see \cite{Du2015b,Zhu2017}).
\begin{definition}[\textbf{Convex roof measure}]
	\label{lemma:C_f_convex_roof}	
	For any function $f \in \F$, the convex roof measure of coherence  $C^{\croof}_{f}: \S(\H)  \to \R$ is defined as
	\begin{equation}
	\label{eq:convex_roof_mixed}
	C^{\croof}_{f}(\rho) = \inf_{\left\lbrace  q_k, \ket{ \psi_k} \right\rbrace_{k = 1}^M  \in \D(\rho) } \sum_{k=1}^M q_k 	 f(\mu(\ket{\psi_k}\bra{\psi_k})).
	\end{equation}
\end{definition}
The convex roof measure $C^{\croof}_{f}$ is a good quantifier of coherence since it satisfies conditions \ref{c1:nonneg}-\ref{c4:convexity}. Moreover, the infimum in \eqref{eq:convex_roof_mixed} can be replaced by a minimum, since there is always an optimal pure state decomposition of $\rho$ that reaches the infimum (see \eg \cite{Uhlmann1998}).

The name of the measure $C^{\croof}_{f}$ is based on the fact that  it is the 
convex roof extension  of any coherence monotone with associated function equal to $f$.  An important property of this measure is the following.

\begin{proposition}
	\label{lemma:convex_roof+convexidad}
	Let $C : \S(\H)  \to \R$ be a coherence measure. Then,
	\begin{equation}
	C \leq C^{\croof}_{f_C},  
	\end{equation}	
	where $f_{C}$ is the 
	associated function of $C$.
\end{proposition}

The convex roof construction is widely used, especially in the context of entanglement measures  \cite{Bennett1996,Vidal2000}. 
However, as we mentioned before, it is not the only way to define a coherence measure from a  function $f \in \F$. 
Recently, an alternative construction was proposed \cite{Yu2020}. 
Before introducing this measure of coherence, we need to define the set of all pure states that can be converted into a state $\rho$ by means of incoherent operations,
\begin{equation}
\label{eq:set_pure_to_rho}
\O(\rho)= \left\{ \ket{\psi}\bra{\psi} : \ket{\psi}\bra{\psi} \io \rho \right\}.
\end{equation}

Now, we introduce the coherence measure given in \cite{Yu2020}. 
In this work we will call it \textit{top monotone of coherence}.

\begin{definition}[\textbf{Top monotone}]
	\label{lemma:C_f_operational}	
	For any function $f \in \F$, the top monotone of coherence $C^{\op}_{f}: \S(\H)  \to \R$ is defined as
	\begin{equation}
	\label{eq:C_f_operational_mixed}
	C^{\op}_{f}(\rho) = \inf_{\ket{\psi}\bra{\psi} \in \O(\rho) } f(\mu(\ket{\psi}\bra{\psi})).
	\end{equation}
\end{definition}

The top monotone $C^{\op}_{f}$ satisfies conditions \ref{c1:nonneg}-- \ref{c5:normalization}, whereas condition \ref{c4:convexity} holds if and only if $C^{\op}_{f} = C^{\croof}_{f}$ \cite[Th.4]{Yu2020}.
The chosen name for this measure is based on the following property given in  \cite{Yu2020}.
\begin{proposition}
\label{prop:relation_top}.
Let $C : \S(\H)  \to \R$ be a coherence monotone. Then,
	\begin{equation}
	\label{eq:relation_top}
	C \leq C^{\op}_{f_C},  
	\end{equation}	
	where $f_{C}$ is the 
	associated function of $C$.
\end{proposition}

As in the case of coherence measures based on the convex roof construction, the infimum in~\eqref{eq:C_f_operational_mixed} can be replaced by a minimum, since there always exists a pure state that reaches the infimum. 
This is a consequence of the continuity of $f$ on $\Delta_d$ (concave functions on $\mathbb{R}^d$ are continuous on any subset of $\mathbb{R}^d$ \cite[Th.10.1]{Rockafellar}) and the compactness of the set $\O(\rho)$, a fact that we will show in Lemma~\ref{lemma:compact}. In some proofs given in \cite{Yu2020} it is assumed the existence of the minimum in~\eqref{eq:C_f_operational_mixed}, but its existence is not prove in general (see for instance the proofs of monotonicity and strong monotonicity of $C^{\op}_{f}$, or the converse part of the proof of  Th. 3 regarding the convexity of $C^{\op}_{f}$, or the proof of Th. 7 regarding the continuity of $C^{\op}_{f}$).  
Therefore, our Lemma~\ref{lemma:compact} fills these gaps.

\section{Generalized Coherent vector: definition and properties}
\label{sec:coherence_vector}

In this section, we introduce the \textit{generalized coherence vector} for arbitrary quantum states. This definition generalizes the one given in \eqref{eq:cohrence_vector_pure}, and it connects the notion of coherence with the majorization lattice theory. Moreover, it allows to introduce a new family of coherence quantifiers, alternative to $C_f^\croof$ and $C_f^\op$.

Inspired by the definitions of $C^{\croof}_{f}$ and $C^{\op}_{f}$, we define two sets of probability vectors associated with a given quantum state $\rho$. 
The first one is obtained from the pure state decompositions of $\rho$. We denote it as $\Ud(\rho)$, where the acronym ``$\mathrm{psd}$'' refers  to \textit{pure state decompositions of $\rho$}. 
\begin{definition}[\textbf{Pure state decompositions set}]
	For any quantum state $\rho$, the pure state decompositions set of $\rho$ is defined as
	\begin{equation}
	\Ud(\rho)= \left\{ \sum_{k = 1}^{M}  q_k \mu^\downarrow(\ket{\psi_k}\bra{\psi_k}) : \{q_k, \ket{\psi_k}\}_{k= 1}^M  \in \D(\rho) \right\}.
	\end{equation}
\end{definition}

The second set of probability vectors associated with a quantum state $\rho$ is obtained from the set of all pure states that can be converted into $\rho$. We denote it as $\Uo(\rho)$, where the acronym ``$\mathrm{psc}$'' refers to \textit{pure states connected to $\rho$}.   

\begin{definition}[\textbf{Connected pure states set}]
	For any quantum state $\rho$, the connected pure states set
	of $\rho$ is defined as
	\begin{equation}
	\Uo(\rho)= \left\{ \mu^\downarrow(\ket{\psi}\bra{\psi}) : \ket{\psi}\bra{\psi} \in \O(\rho)  \right\}.	
	\end{equation}
	
\end{definition}

An interesting property of these sets is that both are convex sets.
\begin{proposition}
	\label{prop:convexity of the sets}		
	The sets $\Ud(\rho)$ and $\Uo(\rho)$ are convex.
\end{proposition}

Another observation that will be useful for characterizing quantum coherence is the following. 
For a given $\rho$, $\Ud(\rho), \Uo(\rho) \subseteq \Delta^\downarrow_{d}$, and, since the majorization lattice is complete (see \eg \cite{Bapat1991,Bosyk2019}), the supremum and infimum (with respect to majorization partial order) of both sets always exist. Moreover, the supremum of both sets coincide.
This result is stated in the following proposition.

\begin{proposition}
	\label{prop:supremum_equality}
	$\bigvee\Ud(\rho) = \bigvee\Uo(\rho)$.
\end{proposition}

This result allows to define the coherence vector of a general quantum state, generalizing the definition given in \eqref{eq:cohrence_vector_pure}. 

\begin{definition}[\textbf{Generalized coherence vector}]
	\label{def:coherence_vector_mixed}
	For any quantum state $\rho$, its generalized coherence vector $\vc(\rho)$ is defined as 
	\begin{equation}
	\label{eq:coherence_vector_mixed}
	\vc(\rho) = \bigvee \Ud(\rho),
	\end{equation}
	or, equivalently as $\vc(\rho) = \bigvee \Uo(\rho)$. 

\end{definition}

Notice that for a pure state, the generalized coherence vector is equal to the ordered coherence vector, i.e., $\vc (\ket{\psi}\bra{\psi}) = \mudown (\ket{\psi}\bra{\psi})$, which means that the Def.~\ref{def:coherence_vector_mixed} is a suitable extension of Def.~\ref{def:coherence_vector_pure}. 
	
We observe that whenever $\bigvee \Ud(\rho) \in \Ud(\rho)$, $\bigvee \Ud(\rho)$ is a maximum.
We call \textit{optimal pure state decomposition} to the ensemble that reaches this maximum.
\begin{definition}[\textbf{Optimal pure sate decomposition}]
	\label{def:optimal pure sate decomposition}
	An ensemble $\{\tilde{q}_k, \ket{\tilde{\psi}_k}\}_{k = 1}^M$ is an optimal pure sate decomposition of $\rho$ if  $\{\tilde{q}_k, \ket{\tilde{\psi}_k}\}_{k = 1}^M \in \D(\rho)$ and $\sum_{k = 1}^{M}   \tilde{q}_k \mu^\downarrow(\ket{\tilde{\psi}_k}\bra{\tilde{\psi}_k}) = \nu(\rho)$. 
	\
\end{definition}
In Ref.~\cite{Yu2020}, it is stated that for a general quantum state it is not easy to prove 
whether an optimal pure state decomposition always exists 
In Sec.~\ref{sec:examples}, we will provide a method to check if the supremum is a maximum. 
In particular, we will show that there are qutrit states for which the optimal ensemble does not exist. This implies that in general the optimal pure state decomposition of a quantum state does not exist.

Whenever $\bigvee \Uo(\rho) \in \Uo(\rho)$, $\bigvee \Uo(\rho)$ is also a maximum.
We call \textit{optimal pure state} to the state that reaches the maximum.
\begin{definition}[\textbf{Optimal pure state}]
	\label{def:optimal pure sate}
	A pure state $\ket{\tilde{\psi}}$ is optimal if $\ket{\tilde{\psi}}\bra{\tilde{\psi}} \in \O(\rho)$ and $ \mu^\downarrow(\ket{\tilde{\psi}}\bra{\tilde{\psi}}) = \nu(\rho)$.
\end{definition}

We have that when there exists an optimal pure state decomposition, there also exists an optimal pure state, and vice versa.
\begin{proposition}
	\label{lemma:maximum_sets}
	$\vc(\rho) \in \Ud(\rho) \iff \vc(\rho) \in \Uo(\rho)$.
\end{proposition}

In what follows, we will show that the generalized coherence vector satisfies several properties that capture the main features of quantum.
The first observation is that the generalized coherence vector completely characterizes the notion of incoherent state.
\begin{proposition}	
\label{prop:incoheren_state}
$\rho$ is incoherent~$\!\iff\!\vc(\rho)\!=\!~(1,0,\ldots,0)$.
\end{proposition}
\noindent This result justifies the Def.~\ref{def:coherence_vector_mixed} for the generalized coherence vector. 

We also have that the generalized  coherence vector fully characterizes maximally coherent states. 
\begin{proposition}
	\label{prop:coherence_vector_maximum_coherence}
	$\rho$ is maximally coherent $\iff \vc(\rho) = \left(\frac{1}{d}, \ldots, \frac{1}{d}\right)$.
\end{proposition}

We observe that, by definition, for any pure state decomposition of $\rho$, the following majorization relation is satisfied.
\begin{proposition}
	\label{prop:convexity_mu}
	Let $\rho = \sum_{k = 1}^{M} p_k \ket{\psi_k}\bra{\psi_k}$. Then,
	\begin{equation}
	\label{eq:convexity_mu}
	\sum_{k = 1}^{M}  p_k \vc(\ket{\psi_k}\bra{\psi_k}) \preceq \vc(\rho).
	\end{equation}
\end{proposition}
This result allows us to interpret the definition  $\nu(\rho) = \bigvee \Ud(\rho) $  as a kind of concave roof extension of the pure state coherence vector.\footnote{Indeed, our proposal can be stated in a more abstract framework, generalizing the notion of concave roof extension of a function. 
More precisely, let $\Omega$ be a compact convex set and $\Omega^{\mathrm{pure}}$ be the set formed by its extremal points.
The concave roof $f^{\croof}:\Omega \to \R$ of the function $f:\Omega^{\mathrm{pure}}\to\R$ is defined as  $f^{\croof}(\omega)=\sup \sum_k q_k f(\omega_k)$, where the supremum is taken over all extremal convex decompositions of $\omega= \sum_k q_k \omega_k, \ \omega_k \in \Omega^{\mathrm{pure}}$ (see \eg \cite{Uhlmann2010}).
This construction can be extended to functions from a compact convex set to the majorization lattice.
The concave roof $\vec{f}^{\croof}:\Omega \to \Delta^\downarrow_{d}$ of the function $\vec{f}:\Omega^{\mathrm{pure}}\to\Delta^\downarrow_{d}$ is defined as $\vec{f}^{\croof}(\omega) = \bigvee \sum_k q_k \vec{f}(\omega_k)$,
where, in this case, the supremum is the one of the majorization lattice.
}
However, the question whether the above statement is valid for any mixed state decomposition of $\rho$ remains open.

Alternatively, the equivalence 
$\bigvee \Ud(\rho)= \bigvee \Uo(\rho)$, gives the generalized coherent vector an operational interpretation in terms of pure state transformations.
In this sense, our definition is 
physically and mathematically well motivated, and it is a 
suitable extension of the
pure state coherence vector given in Def.~\ref{def:coherence_vector_pure}.

\section{Necessary conditions for incoherent transformations}
\label{sec:necessary conditions}

In this section, we apply the notion of generalized coherence vector, given in  \ref{def:coherence_vector_mixed}, for characterizing state transformations between arbitrary quantum states.

\begin{proposition}
	\label{prop:monotonocity_vectorcoherence}
	Let $\rho$ and $\sigma$ be two arbitrary quantum states. Then,
	\begin{multline}
	\label{eq:monotonocity_vectorcoherence}
	\rho \io \sigma \! \implies \!  \forall \{q_k,\ket{\psi_k}\}_{k = 1}^M \!\in\! \D(\rho),   \exists \{r_l,\ket{\phi_l}\}_{l \in L} \!\in\! \D(\sigma): 	\sum_{k = 1}^{M}  q_k \mu^\downarrow(\ket{\psi_k}\bra{\psi_k}) \preceq  \sum_{l \in L}  r_{l} \mu^\downarrow(\ket{\phi_{l}}\bra{\phi_{l}}). 
	\end{multline}
\end{proposition}
Notice that this result generalizes the necessary condition of Prop.~\ref{lemma:theorem4du2019}.
In addition, we have the following consequences.

\begin{corollary}
	\label{cor:monotonocity_vectorcoherence 2}
	Let $\rho$ and $\sigma$ be two  quantum states. Then, 	
	\begin{equation}
	\label{eq:strong_monotonocity_vectorcoherence_supremum}
	\rho \io \sigma  \implies \vc(\rho) \preceq \sum_{n = 1}^{N}  p_n \vc(\sigma_n),
	\end{equation}
	with $p_n = \Tr\left(K_n \rho K^\dag_n\right)$ and $\sigma_n = K_n \rho K^\dag_n/p_n$, where $\{K_n\}_{n = 1}^N$ are incoherent Kraus operators such that $\sigma= \sum_{n = 1}^{N}  K_n \rho K^\dag_n$.
\end{corollary}

We observe that the majorization relation~\eqref{eq:strong_monotonocity_vectorcoherence_supremum} generalizes the necessary condition for incoherent transformations from pure to arbitrary states, given in~\eqref{eq:lemma4zhu}, to the general case, i.e., from arbitrary states to arbitrary states.

Another consequence of Prop.~\ref{prop:monotonocity_vectorcoherence} is that the generalized  coherence vectors of two states $\rho$ and $\sigma$ satisfy a majorization whenever $\rho$ can be transformed into $\sigma$.
\begin{corollary}
\label{cor:monotonocity_vectorcoherence_supremum}
Let $\rho$ and $\sigma$ be two quantum states. Then,
\begin{equation}
\label{eq:monotonocity_vectorcoherence_supremum}
\rho \io \sigma  \implies \vc(\rho) \preceq \vc(\sigma).
\end{equation}
\end{corollary}

Notice that the r.h.s
condition is not sufficient even for  qubit systems. 
In fact, a qubit state $\rho$ with Bloch vector $(r_x,r_y,r_z)$ can be converted into another state $\sigma$ with Bloch vector $(s_x,s_y,s_x)$ by means of incoherent operations if and only if two conditions are satisfied: (i) $s_x^2 +s_y^2 \leq r_x^2 +r_y^2 $, and (ii) $s_z^2 \leq 1- (1-r_z^2)/(r_x^2 +r_y^2) (s_x^2 +s_y^2) $ (see \cite{Streltsov2017b,Shi2017}). 
By using the result given in Eq.~\eqref{eq:maximum qubit} (or in \cite{Yu2020}), it can be shown that only condition (i) is equivalent to the r.h.s of~\eqref{eq:monotonocity_vectorcoherence_supremum}. 
Moreover, in higher dimensions ($d \ge 4$),  a finite number of conditions in terms of coherence measures are not enough to completely characterizes the coherence transformations \cite{Du2019}.


\section{A family of coherence monotones}
\label{sec:familiy_monotones}

In this section, we introduce a new family of coherence monotones, alternative to $C_f^\croof$ and $C_f^\op$. 
We adopt a different approach to the ones given in Def. \ref{lemma:C_f_convex_roof} and \ref{lemma:C_f_operational}.
Our proposal is based on the generalized coherence vector introduced in Def.~\ref{def:coherence_vector_mixed}. 
The fact that this definition satisfies the properties given in Prop.~\ref{prop:incoheren_state}--\ref{prop:monotonocity_vectorcoherence} allows us to introduce the following family of coherence quantifiers, which we call \textit{coherence vector monotone}. 

\begin{definition}[\textbf{Coherence vector monotone}]
	\label{def:C_f_sup}	
	For any function $f \in \F$, the coherence vector monotone $C^{\cv}_{f}: \S(\H)  \to \R$ is defined as
	\begin{equation}
	\label{eq:C_f_sup}
	C^{\cv}_{f}(\rho) =  f\left( \vc(\rho)\right), 
	\end{equation}
where $\vc(\rho)$ is the generalized coherence vector of $\rho$.
	
\end{definition}

We observe that this family of quantifiers is well defined. 
The following result states that it satisfies the first four conditions of a coherence measure.

\begin{proposition}
	\label{prop:cohrence_measures_pure_conv}	
		For any function $f \in \F$, the coherence vector monotone $C^{\cv}_{f}$ satisfies conditions \ref{c1:nonneg}--\ref{c5:normalization}. 
\end{proposition}
\noindent We observe that	a coherence vector monotone $C_f^{\cv}$ can only be convex if $C_f^{\cv} \leq C_f^{\croof}$, as in the case of Eq.~\eqref{eq:geometric}. 

We stress that any function $f \in \F$ gives a coherence monotone. 
In others words, the function $f$ can be arbitrarily chosen from the set $\F$ as in the cases of the convex roof measures and the top monotones.

In what follows, we are going to characterize the order relation among the coherence quantifiers $C^{\croof}_f$ , $C^{\op}_f$ and $C^{\cv}_f$.
First, we note that, due to Prop.~\ref{prop:relation_top}, 
$C^{\op}_f \geq C^{\croof}_f$  and $C^{\op}_f \geq C^{\cv}_f$.
Moreover, for some $\rho \in \S(\H)$, we have $C^{\op}_f(\rho) = C^{\cv}_f(\rho)$. 
The following result characterizes this situation. 

\begin{proposition}
	\label{prop:relation_coherence_measures_op_sup}
The following statements are equivalent: 
\begin{enumerate}
\item  There exists an optimal pure state decomposition of $\rho$, i.e., $\vc(\rho) \in \Ud(\rho)$. \label{st1}
\item $C^{\cv}_f(\rho) =  C^{\op}_f(\rho)$ for all $f \in \F$. \label{st2} 
\item $C^{\cv}_f(\rho) =  C^{\op}_f(\rho)$ for some $f \in \F$, with $f$ strictly Schur-concave. \label{st3}
\end{enumerate}
\end{proposition}
This result gives us a method to address the question about the existence of an optimal pure state decomposition 
of a general quantum state.
In Sec.~\ref{sec:examples} we will use it to show that for some qutrit states there exist an optimal pure state decomposition.

In general, there is not a defined order relation between $C^{\cv}_f$ and $C^{\croof}_f$.
However, when there exists an optimal pure state decomposition of a state, we have the following result.
\begin{proposition}
	\label{prop:relation_coherence_measures_sup_croof}
If there exists an optimal pure state decomposition of $\rho$, then $C^{\cv}_f(\rho) \geq C^{\croof}_f(\rho)$.
\end{proposition}
On the contrary, for affine functions, we have the opposite relation between $C^{\cv}_f$ and $ C^{\croof}_f$.

	\begin{proposition}
		\label{prop:linear f}
	Let $f\in\F$ be such that $f|_{\Delta_d^\downarrow}=c+\ell$, where $c\in\mathbb{R}$ and $\ell : \Delta_d^\downarrow \to \mathbb{R}$
	is a linear function. 
	Then, $C^{\cv}_f \leq C^{\croof}_f=C^{\op}_f$.
		\end{proposition}

Examples of this class of functions are $f(u)=1-u_1^\downarrow$, $f(u)=u_d^\downarrow$ and $f(u)=1-u_1^\downarrow+u_d^\downarrow$, where $u^\downarrow_i = (u^\downarrow)_i$. 
In particular, for the former function, we have that all quantifiers coincide and are equal to the geometric measure of coherence \cite{Streltsov2015}, i.e.,  for the function  $f(u)=1-u_1^\downarrow$ we have
\begin{equation}
\label{eq:geometric}
C^{\cv}_f (\rho)= C^{\croof}_f (\rho)=C^{\op}_f(\rho)   = \min_{\left\lbrace  q_k, \ket{ \psi_k} \right\rbrace_{k= 1}^ M  \in \D(\rho)}  \sum^M_{k=1} q_k \left( 1-\max_{0 \leq i \leq d-1} |\braket{i|\psi_k}|^2\right) .
\end{equation}

Moreover, whenever there exists the optimal pure state decomposition of $\rho$, $C^{\cv}_f(\rho)$ and $C^{\croof}_f(\rho)$ are equal for a subclass of functions of  $\F$.
\begin{proposition}
	\label{prop:exsitence_optimalpsd_roof}
If there exists an optimal pure state decomposition of $\rho$, i.e., $\vc(\rho) \in \Ud(\rho)$, then
	\begin{equation}
		\label{eq:exsitence_optimalpsd_roof}
		C^{\cv}_{f_k}(\rho) =  C^{\croof}_{f_k}(\rho),  
	\end{equation}
with $f_k(u) = 1-\sum^{k-1}_{i=0} u^\downarrow_i \in \F$, for all $k \in \{1, \ldots, d-1\}$.	
\end{proposition}

The scheme of Fig.~\ref{fig:scheme} summarizes the relationships among the three families of coherence quantifiers.

\begin{figure}
\includegraphics[width=0.35\textwidth]{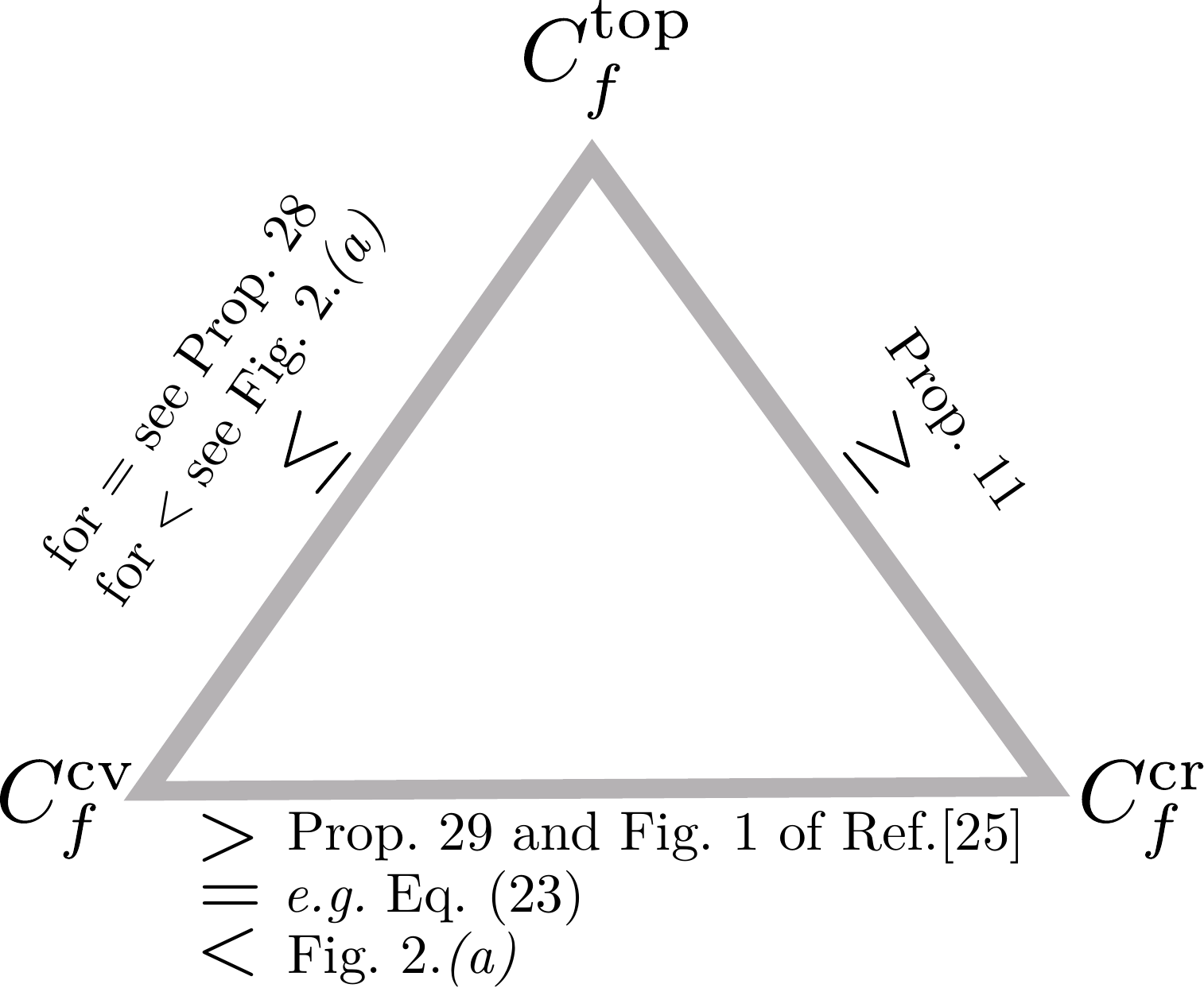}
\caption{Scheme of the relationships among $C^{\croof}_f$, $C^{\op}_f$ and $C^{\cv}_f$.}
\label{fig:scheme}
\end{figure}

\section{Examples}
\label{sec:examples}

In this section, we calculate the coherence vector for two simple models. First, we consider a qubit state, and we obtain its generalized coherence vector. 
Secondly, we consider a maximally coherent qutrit going through a depolarizing channel, and we compute the value of $C_f^{\croof}(\rho)$, $C_f^\op(\rho)$ and $C_f^{\cv}(\rho)$ for a given state $\rho$ and $f \in \F$.

\subsection{Qubit case}

Let us consider a qubit system in a state $\rho = \frac{1+ \vec{r} \cdot\vec{\sigma} }{2}$, with  $\vec r=(r_x,r_y,r_z)$ the Bloch vector ($\|\vec r \| \leq 1$),  and $\vec{\sigma} =(\sigma_x,\sigma_y,\sigma_z)$ the vector formed by the Pauli matrices.

In a previous work \cite{Yu2020}, it has been shown that the supremum of $\Ud(\rho)$ is a maximum, and it is given by
\begin{equation}
	\label{eq:maximum qubit}
	\vc(\rho) =\left(\frac{1+ r}{2},  \frac{1 - r }{2} \right),  
\end{equation}
where $r= \sqrt{1- r^2_x-r^2_y}$.
An optimal pure state decomposition of $\rho$ is given by $\{q,\ket{\psi^+}; 1-q,\ket{\psi^-} \}$ where
$\ket{\psi^\pm}\bra{\psi^\pm} = \frac{1+ \vec{s}^{\pm} \cdot\vec{\sigma} }{2},$ with $\vec{s}^{\pm} =(r_x,r_y,\pm r)$ and $ q = (r_z+r)/2r \in [0,1]$.  
As a consequence of Prop.~\ref{prop:relation_coherence_measures_op_sup}, we have that $C^{\cv}_f= C^\op_f$. 
	
Furthermore, it has been shown that for any function $f( \vc ( \rho))  = \tilde{f}(r)$, such that $\tilde{f}$ is a convex function on $r$, $C^\op_f$ is a convex monotone of coherence and $C^\op_f = C^\croof_f$ \cite{Yu2020}.
For the qubit case, most of the well-known coherence measures, like $\ell_1$-norm, relative entropy, geometric coherence admit a formulation in terms of a convex function of $r$. 
This means that we have the triple equivalence among the families $C^{\cv}_f$, $C^\op_f$ and $C^{\croof}_f$ in this case.

Due to Prop.~\ref{prop:relation_coherence_measures_op_sup}, to observe a difference between  $C^{\cv}_f$ and $C^\op_f$, we need an example where $\vc(\rho)$ is not a maximum.  
This could be possible, in principle, in higher dimensions ($d\geq 3$).
In what follows, we provide an example for $d=3$.

\subsection{Qutrit case}

Let us consider a qutrit system in the  maximally coherent state $\ket{\psi^{\mcs}} = (\ket{0} + \ket{1}+\ket{2})/\sqrt{3}$,
going through a depolarizing channel with depolarization probability $p$. The final state after the depolarizing channel is given by
\begin{equation}
\label{ej:depol_3}
\rho_p  = \Lambda_p\left(\ket{\psi^{\mcs}}\bra{\psi^{\mcs}}\right) = p \, \frac{I}{3} + (1-p) \, \ket{\psi^{\mcs}}\bra{\psi^{\mcs}}.
\end{equation}

Also, we consider the function $f(u)=1-u_1^\downarrow+u_d^\downarrow$.
Clearly, $f$ satisfies the conditions of Prop. \ref{prop:linear f}.
Therefore, we have that for this function both measures $C_f^{\croof}$ and $C_f^\op$ are equal.

In Fig.\ref{fig:ej1}.(a), we plot $C_f^{\croof}(\rho_p)$ (or, equivalently $C_f^\op(\rho_p)$) and $C_f^{\cv}(\rho_p)$ as functions of $p \in [0,1]$.
Both functions are monotonically decreasing in terms of $p$, and in the open interval $(0,1)$, we have $C_f^{\croof}(\rho_p) = C_f^\op(\rho_p) > C_f^{\cv}(\rho_p)$. 
Equivalently,  this means that the supremum is not a maximum (see Prop.~\ref{prop:relation_coherence_measures_op_sup}). 
In Fig.\ref{fig:ej1}.(b), we plot $\vc(\rho_p)$ and $u^{\op}_f(\rho_p) = \argmin_{u \in \Uo(\rho)} f(u)$.
It is shown that $u^{\op}_f(\rho_p) \preceq \vc(\rho_p)$ and $u^{\op}_f(\rho_p) \neq \vc(\rho_p)$ for several values of $p$ in the open interval $(0,1)$.
Finally, in Fig.\ref{fig:ej1}.(c), we consider $\rho_p$ for $p=0.3$ and we depict $\vc(\rho_p)$ and the region $\{u \in  \Delta^\downarrow_3: u \preceq \vc(\rho_p) \}$. 
In addition, we generate $10^5$ random unitary  from $3$ up to $9$ \footnote{
The upper bound $9$ is not arbitrary. According to Lemma 1 in
\cite{Uhlmann1998}, the optimal $C_f^{\croof}(\rho_p)$ 
requires at most nine terms.
It is conjectured in \cite[Conjecture and Lemma 7]{Uhlmann1996}
that three terms are enough.}
.
For each unitary matrix, we use the Schr\"ordinger theorem (see Eq.\eqref{eq:SchTh} or \cite{Nielsen2000}) to generate an ensemble compatible with $\rho_p$  and we plot its coherence vector. 
This plot depicts that the optimal pure state decomposition of this state does not exist (which can be inferred from the left figure and Prop.~\ref{prop:relation_coherence_measures_op_sup}).

\begin{figure*}[!htb]
\includegraphics[width=\textwidth]{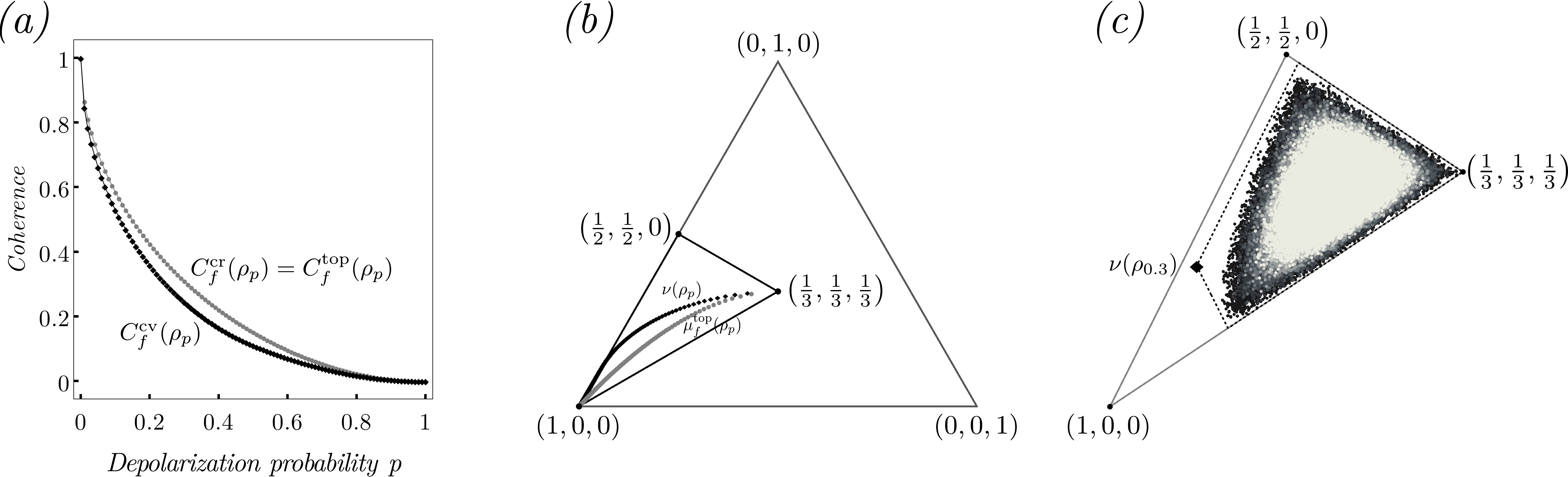}
\caption{Let $\rho_p=\Lambda_p\left(\ket{\psi^{\mcs}_3}\bra{\psi^{\mcs}_3}\right)$ and $f(u)=1-u_1^\downarrow+u_d^\downarrow$. 
	\textit{(a)} $C_f^\croof(\rho_p)$ (or $C_f^\op(\rho_p)$) (gray circle) and $C_f^{\cv}(\rho_p)$ (black diamond) in function of depolarization probability $p$.  
	\textit{(b)} The big triangle represents the  $\Delta_3$, whereas that the small triangle depicts the set $\Delta^\downarrow_{3}$. 
	We plot  $\vc(\rho_p)$ (black diamond) and $u^{\op}_f(\rho_p) = \argmin_{u \in \Uo(\rho)} f(u)$ (gray circle) for several values of $p \in [0,1]$. 
	\textit{(c)} For $p=0.3$, we plot $\vc(\rho_p) \approx (0.777,0.2,0.0223)$ (black diamond) and the region $\{u \in  \Delta^\downarrow_3: u \preceq \vc(\rho_p) \}$ (enclosed by the black dashed lines). In addition, we generate $10^5$ random unitary matrices of dimensions from $3$ to $9$. For each unitary matrix $U$, we use the Schr\"ordinger theorem (see \eg \cite{Hughston1993,Nielsen2000}) to generate an ensemble compatible with $\rho_p$ ($p=0.3$), i.e., $\sqrt{p_k} \ket{\phi_k} = \sum_{j=1}^{3} \sqrt{\lambda_j} U_{k,j} \ket{e_j}$, where $\lambda_j$ and $\ket{e_j}$ are the eigenvalues and eigenstates of $\rho_p$. 
	For each ensemble, we plot its coherence vector in grayscale representing the dimension of the unitary matrix (from 3 to 9) used to generate the ensemble. 
	The darkest gray corresponds to dimension 3, whereas the lightest gray corresponds to dimension 9.
	This plot evidences that the optimal set of this state does not exist (which can be inferred from the left figure and Prop.~\ref{prop:relation_coherence_measures_op_sup}).}
\label{fig:ej1}
\end{figure*}

\section{Concluding Remarks}
\label{sec:conclusions}

In this work, we have advanced on  the characterization of the quantum coherence resource theory by defining the generalized coherence vector of an arbitrary quantum  state.
This probability vector can be interpreted as a concave roof extension of the coherence vector defined for pure states.
We showed that it is a good definition, since it allows to characterize the notions of being incoherent, as well as being maximally coherent. 
Using this definition and the majorization relation, we obtain a necessary condition for the conversion of general quantum states by means of incoherent operations.
This generalizes the result for pure states given in the literature, and shows that the tools of the majorization lattice are useful also in the general case.

Moreover, based on the generalized coherence vector, we introduced a family of monotones, called coherence vector monotones. 
In order to do this, we considered concave and symmetric functions applied to the generalized coherence vector of a quantum state. 
In this way, our approach is an alternative method to construct extend coherence measures from pure to mixed states. 
This family of monotone was compared with the families of the convex roof measure and the top monotone.
We obtain that the coherence vector monotone is lower than or equal to the top monotone, and the equality is only satisfied when  the generalized coherence vector of the state is a maximum.  
In addition, we have obtained that there is not a definite order between the convex roof measure and the coherence vector monotone.
We provided several examples showing that  our quantifier can be strictly greater than, equal to or strictly lower than the convex roof measure.
We have also applied the coherence vector monotone to quantify the coherence of a qubit system and a maximally coherent qutrit going through a depolarizing channel. 

Finally, we stress that our framework, which is mainly based on the majorization lattice theory, could also be used in other majorization-based resource theories.
Moreover, it would be interesting to study whether this approach can be extend to more general resource theories of coherence, such as the ones  based on non-orthonormal basis or on positive-operator-valued measures.

\acknowledgments

This work is partially supported by the projects ``Indagine sulla struttura logica e geometrica soggiacente alla teoria dell’informazione
quantistica'' funded by Ministero dell'Università e della Ricerca (Italy), ``Strategies and Technologies for Scientific Education and Dissemination'' funded by the Fondazione di Sardegna (Italy), 
``Per un’estensione semantica della Logica Computazionale Quantistica - Impatto teorico e ricadute implementative'' (RAS: SR40341, L.R. 7/2007) funded by Regione Autonoma della Sardegna (Italy), PIP-0519 funded by Consejo Nacional de Investigaciones Científicas y Técnicas  (Argentina) and PICT-01774 funded by Agencia Nacional de Promoción Científica y Tecnológica (Argentina).

\appendix

\section{Auxiliaries lemmas}
\label{sec: App1}

The following result is necessary for Prop.~\ref{prop:monotonocity_vectorcoherence}. It  states that any convex combination of ordered probability vectors preserves the majorization relation. 

\begin{lemma}
	\label{lemma:resultado_auxiliar}
	Let $u_0, \ldots, u_m \in \Delta_{d}^\downarrow$ and $v_0, \ldots, v_m \in \Delta_{d}^\downarrow$ be two sequence of ordered probability vectors, such that  $u_\ell \preceq v_\ell$, for all $0 \leq \ell  \leq m$. 
	For any probability vector $q = (q_0, \ldots, q_{m}) \in \Delta_{m+1}$, the vectors $u = \sum_{\ell=0}^{m} q_\ell u_\ell$ and $v = \sum_{\ell=0}^{m} q_l v_l$ belong to $\Delta_{d}^\downarrow$, and $ u  \preceq v$.
	
\end{lemma}

\begin{proof}
	Let $q = (q_0, \ldots, q_m)$ be an arbitrary probability vector in $\Delta_{m+1}$. 
	Firstly, we note that $(u)_i =  \sum_{\ell=0}^{m} q_\ell \left(u_\ell\right)_i \geq 0$ for all $0\leq i \leq d-1$, and $\sum_{i=0}^{d-1} (u)_i = \sum_{i=0}^{d-1}\sum_{\ell=0}^{m} q_\ell (u_\ell)_{i} =\sum_{\ell=0}^{m} q_\ell \sum_{i=0}^{d-1} (u_\ell)_{i} =1$, i.e., $u \in \Delta_{d}$. Moreover,  since $(u_\ell)_{i+1} \leq (u_\ell)_{i}$ for all $0\leq i \le d-2$ and for all $0\le \ell \le m$, 
	we have	$(u)_{i+1} = \sum_{\ell=0}^{m} q_\ell (u_\ell)_{i+1} \leq \sum_{\ell=0}^{m} q_\ell (u_\ell)_{i} = (u)_{i}$. 
	Hence, $u \in \Delta_{d}^\downarrow$. 
	Analogously, $v \in \Delta_{d}^\downarrow$.
	
	Secondly, since $u_\ell \preceq v_\ell$,  for all $0\leq \ell \leq m$, then we have  
	$\sum_{i=0}^{k}  (u_\ell )_{i}\leq \sum_{i=0}^{k}  (v_\ell )_{i}$, for all $0 \leq  k \leq d-1$.
	Therefore, for all $0 \leq  k \leq d-1$, we have $\sum_{i=0}^{k} (u)_i =\sum_{i=0}^{k}  \sum_{\ell=0}^{m} q_\ell  (u_\ell )_{i} =   \sum_{\ell=0}^{m} q_\ell \sum_{i=0}^{k} (u_\ell )_{i}  \leq \sum_{\ell=0}^{m} q_\ell \sum_{i=0}^{k} (v_\ell )_{i} = \sum_{i=0}^{k} (v)_i$. Hence,  $ u  \preceq v$.
	
\end{proof}

The following result is necessary for  Prop~\ref{prop:relation_coherence_measures_op_sup}.

\begin{lemma}
	\label{lemma:strict_schur}
	Let $f: \Delta_d \to \mathbb{R}$ be a strictly Schur-concave function, and $u,v \in \Delta_d$. 
	If $f(u) = f(v)$, then either (i) $u \npreceq v$ and $v \npreceq u$ (incomparable) or (ii) $u = \Pi v$, with $\Pi$ a permutation matrix. 
\end{lemma}

\begin{proof}
	Given $u, v \in \Delta_d$, we suppose that $f(u) = f(v)$. 
	Then, there are two options: (i) $u$ and $v$ are incomparable or (ii) $u$ and $v$ are comparable.
	If (i) is the case, there is nothing to prove. 
	If (ii) is the case, without loss of generality, we can assume $u \preceq v$.
	Since $f(u) = f(v)$, and $f$ is strictly Schur-concave, we conclude $u = \Pi v$, with $\Pi$ a permutation matrix.
\end{proof}

The following two lemmas will be necessary to prove that the sets $\Ud(\rho)$ and $\Uo(\rho)$ have the same supremum (see Prop.\ref{prop:supremum_equality}).

\begin{lemma}
	\label{lemma:inclusion}
	$\Ud(\rho)\subseteq\Uo(\rho)$. 
\end{lemma}

\begin{proof}
	Given an arbitrary $u \in \Ud(\rho)$, there exists a pure state decomposition $\{q_k, \ket{\psi_k}\}_{k = 1}^M$ of $\rho$ such that $\sum_{k= 1}^{M} q_k \mu^\downarrow(\ket{\psi_k}\bra{\psi_k}) = u$.
	Moreover, always exists a pure state $\ket{\psi}\bra{\psi}$ such that $\mu^\downarrow(\ket{\psi}\bra{\psi}) = u$.
	
	Since the majorization relation is reflexive, $\mu^\downarrow(\ket{\psi}\bra{\psi}) \preceq  \sum_{k} q_{k} \mu^\downarrow(\ket{\psi_{k}}\bra{\psi_{k}})$. Finally, from Proposition~\ref{lemma:theorem4du2019}, we have that $\ket{\psi}\bra{\psi} \io \rho $, and  $\mu \in \Uo(\rho)$. Therefore,  	$\Ud(\rho)\subseteq\Uo(\rho)$. 
	
\end{proof}

\begin{lemma}
	\label{lemma:cota}
	For each $ u \in \Uo(\rho)$, there exists an element $u' \in \Ud(\rho)$, such that $u\preceq u'$.
\end{lemma}

\begin{proof}
	Given an arbitrary $u \in \Uo(\rho)$,
	there exists a pure state $\ket{\psi}\bra{\psi}$ such that $\mu^\downarrow(\ket{\psi}\bra{\psi}) = u$ and $\ket{\psi}\bra{\psi} \io \rho $. 
	From Proposition~\ref{lemma:theorem4du2019}, there exists a pure state decomposition $\{q_k, \ket{\psi_k}\}_{k = 1}^{M}$ of $\rho$ such that $u \preceq \sum_{k= 1}^{M} q_k \mu^\downarrow(\ket{\psi_k}\bra{\psi_k})$. 
	Since $\rho = \sum_{k= 1}^{M} q_k \ket{\psi_k}\bra{\psi_k}$, then $u' =  \sum_k q_k \mu^\downarrow(\ket{\psi_k}\bra{\psi_k})$ belongs to $\Ud(\rho)$. 
	Therefore, $u \preceq u'$.  
	
\end{proof}

The next result shows the compactness of the set $\O(\rho)$.
This result will be used for the proof of Prop.~\ref{prop:incoheren_state}. 
In addition, Lemma~\ref{lemma:compact} together with the continuity of $f$ allows us to replace the infimum in \eqref{eq:C_f_operational_mixed} by a minimum. In addition, this allows us to fill the gaps of some proofs given in \cite{Yu2020}, where the existence of an optimal state in~\eqref{eq:C_f_operational_mixed} is assumed, but not proved.

\begin{lemma}
	\label{lemma:compact}
	The set $\O(\rho) = \{\ket{\psi}\bra{\psi}: \ket{\psi}\bra{\psi} \io \rho \}$ is a compact set. 
\end{lemma}

\begin{proof}
According to \cite{Streltsov2017b}, there exists a fixed $N$ such that any pure state $\ket{\psi}\bra{\psi} \in \O(\rho)$ satisfies $\sum^{N}_{n=1} K_n  \ket{\psi}\bra{\psi} K^\dag_n = \rho$, where $K_n$ are incoherent Kraus operators. 
By Def.~\ref{def:IO}, the Kraus operators satisfy the two following conditions:
\begin{align}
&\sum_{n} K_n^{\dagger} K_n = I \label{c1:normalization}. \\
&K_n \ket{i} \propto  \ket{f_n(i)}, \text{with} ~ f_n ~ \text{a relabeling of} ~ \{0,\ldots, d-1\}. \label{c2:io_condition} 
\end{align}
On the one hand, from Eq. \eqref{c1:normalization}, it follows that each incoherent Kraus operator is bounded, i.e, $\| K_n \|_{HS} = \Tr( K^\dag_n K_n) \leq d$. 
On the other hand, condition \eqref{c2:io_condition} is equivalent to 
\begin{equation}\label{eq: reescritura_kraus}
(K_n)_{j,i} = (K_n)_{j,i} \, \delta_{j, f_n(i)},  ~~~ \forall i,j \in \{0,\ldots, d-1\}.
\end{equation}
Notice that condition \eqref{eq: reescritura_kraus} are $d^2$ equations for the entries of $K_n$. 

We denote  the set of all relabeling functions as 
$\mathcal{R} = \{f:\{0, \ldots,d-1\} \to \{0, \ldots ,d-1\} \}$ and the $N$-Cartesian product as $\mathcal{R}^N$.  
Given $\vec{f}=(f_1, \ldots, f_N) \in \mathcal{R}^N$, we define the set 
\begin{multline}
K_{\vec{f}} =  \{ \left( K_1, \ldots , K_{N}  \right) \in   \C^{d \times d} \times \ldots \times \C^{d \times d}  :  
\sum_{n=1}^{N} K_n^{\dagger} K_n = I, \    (K_n)_{j,i} = (K_n)_{j,i} \, \delta_{j, f_n(i)}  
\forall i,j \in \{0,\ldots, d-1\}  \},
\end{multline}
and the set
\begin{multline}
V_{\vec{f}}(\rho) = \{ \left( \ket{\psi}\bra{\psi}, K_1, \ldots , K_{N}  \right) \in  \P(\H) \times \C^{d \times d} \times \ldots \times \C^{d \times d}  : 
\left( K_1, \ldots , K_{N}  \right) \in K_{\vec{f}} \, ,  \ \sum^{N}_{n=1} K_n  \ket{\psi}\bra{\psi} K^\dag_n = \rho \}.
\end{multline}
Finally, we consider the set $V(\rho)= \bigcup_{\vec{f} \in \mathcal{R}^N} V_{\vec{f}}(\rho)$.
Since $\mathcal{R}$ is a finite set, $V(\rho)$ is a finite union of sets. Notice that 
\begin{multline}\label{eq:equivalencia}
\ket{\psi}\bra{\psi} \in \O(\rho) 
\iff \exists \left( K_1, \ldots , K_{N}  \right)  \in \C^{d \times d} \times \ldots \times \C^{d \times d} :  
\left( \ket{\psi}\bra{\psi}, K_1, \ldots , K_{N}  \right) \in V(\rho) 
\end{multline}

Since $\P(\H)$ is closed and the set $V_{\vec{f}}(\rho)$ is given by a finite number of equations\footnote{For any continuous function $h$ the set $\{x: h(x)=0\}$ is closed.}, then we have that $V_{\vec{f}}(\rho)$ is a closed set.
Moreover, $V_{\vec{f}}(\rho)$ is bounded, since $\P(\H)$ is bounded and each incoherent Kraus operator has $\| K_n \|_{HS} \leq d$.
Therefore, $V(\rho)$ is a compact set, since it is a finite union of compact sets. 

Let us denote the projection of the set $\P(\H) \times \C^{d \times d} \times \ldots \times \C^{d \times d}$ onto the first coordinate as $\Pi: \P(\H) \times \C^{d \times d} \times \ldots \times \C^{d \times d} \to \P(\H)$.
Since $\Pi$ is a continuous function and $V(\rho)$ is compact, then $\Pi \left( V(\rho)\right) $ is compact. 

We are going to show that $\O(\rho) =\Pi (V(\rho))$. On the one hand, 
let $\ket{\psi}\bra{\psi} \in \Pi \left( V(\rho)\right) $. 
Then, there is an element $(\ket{\psi}\bra{\psi}, K_1, \ldots, K_N) \in V(\rho)$.
Therefore, using equivalence \eqref{eq:equivalencia}, we have 
$\ket{\psi}\bra{\psi} \in \O(\rho)$.
On the other hand, if $\ket{\psi}\bra{\psi} \in \O(\rho)$, there exists $\left( K_1, \ldots , K_{N}  \right)  \in \C^{d \times d} \times \ldots \times \C^{d \times d}$ such that 
$\left( \ket{\psi}\bra{\psi}, K_1, \ldots , K_{N}  \right) \in V(\rho)$. Then,   $\ket{\psi}\bra{\psi} \in \Pi(V(\rho))$.
Therefore, we conclude that $\O(\rho) = \Pi (V(\rho))$ and it is a compact set.
		
\end{proof}

\section{Proofs of propositions given in the main text}
\label{sec: App2}

For the sake of readability, we repeat the statements of the propositions given in the main text and we provide their corresponding proofs.

\begin{repproposition}{lemma:convex_roof+convexidad}
	Let $C : \S(\H)  \to \R$ be a coherence measure. Then,
	\begin{equation}
	C \leq C^{\croof}_{f_C},  
	\end{equation}	
	where $f_{C}$ is a function associated to $C$.
\end{repproposition}

\begin{proof}
	Given an arbitrary quantum state $\rho$, we consider a pure state decomposition $\{q_k, \ket{\psi_k}\}_{k = 1}^M$ of the state, i.e., $\rho = \sum_{k= 1}^{M} q_k   \ket{\psi_k}\bra{\psi_k}$. 
	Since $C : \S(\H)  \to \R$ satisfies  conditions \ref{c1:nonneg}--\ref{c5:normalization}, from Prop. \ref{lemma:cohrence_measures_pure}, there exists a function $f_{C} \in \F$, such that 
	\begin{equation}
	C(\ket{\psi}\bra{\psi}) = f_{C}(\mu(\ket{\psi}\bra{\psi}),  \  \ \ \forall \ket{\psi}\bra{\psi} \in \P(\H).
	\end{equation}
	In addition, $C$ satisfies condition \ref{c4:convexity}, hence 
	\begin{equation}\label{eq:inequality1 lemma4}
	C(\rho)  \leq \sum_{k= 1}^{M} q_k C\left( \ket{\psi_k}\bra{\psi_k}\right) =  \sum_{k= 1}^{M} q_k f_C \left( \mu(\ket{\psi_k}\bra{\psi_k}\right) .
	\end{equation}
	The inequality \eqref{eq:inequality1 lemma4} is valid for any pure state decomposition of $\rho$, then
	\begin{equation}
	\label{eq:inequality2 lemma4}
	C(\rho) \leq  \inf_{\left\lbrace  q_k, \ket{ \psi_k} \right\rbrace_{k = 1}^M  \in \D(\rho) } \sum_{k= 1}^M q_k 	f_C \left( \mu(\ket{\psi_k}\bra{\psi_k}\right).
	\end{equation}
	By definition, the r.h.s of \eqref{eq:inequality2 lemma4} is the convex roof measure for the function $f_C$. Therefore, we obtain $	C(\rho) \leq C^{\croof}_{f_C}(\rho)$, for all $\rho$.
	
\end{proof}

\begin{repproposition}{prop:convexity of the sets}		
	The sets $\Ud(\rho)$ and $\Uo(\rho)$ are convex.
\end{repproposition}

\begin{proof}
We start with the set $\Ud(\rho)$. 
Let $u, u' \in \Ud(\rho)$. Given $t \in \left(0, 1 \right)$, we consider the ordered probability vector $u_t = t u  + (1-t) u'$. 

By definition of $\Ud(\rho)$, we have ${\left\lbrace  q_k, \ket{ \psi_k} \right\rbrace_{k = 1}^M}$ 
and ${\left\lbrace  q'_k, \ket{ \psi'_k} \right\rbrace_{k = 1}^{M'}}$,  two pure state decompositions  of $\rho$, such that $u = \sum_{k= 1}^M q_k \mudown \left( \ket{\psi_{k}}\bra{\psi_{k}}\right) $ and  $u'= \sum_{k=1}^{M'} q'_k   \mudown \left( \ket{\psi'_{k}}\bra{\psi'_{k}}\right) $.
Since $\rho = \sum_{k= 1}^M q_k \ket{\psi_{k}}\bra{\psi_{k}} =  \sum_{k=1}^{M'} q_k' \ket{\psi_{k}'}\bra{\psi_{k}'}$, we have
\begin{equation}
t\sum_{k= 1}^M q_k \ket{\psi_{k}}\bra{\psi_{k}} + (1-t)
\sum_{k=1}^{M'} q_k' \ket{\psi_{k}'}\bra{\psi_{k}'} = \rho.
\end{equation}
Therefore, the join ${\left\lbrace t q_k, \ket{ \psi_k} \right\rbrace_{k = 1}^M} \cup {\left\lbrace  (1-t)q'_k, \ket{ \psi'_k} \right\rbrace_{k = 1}^{M'}}$ is also a pure state decomposition of $\rho$, and $u_t =   \sum_{k= 1}^M  t q_k \mudown \left( \ket{\psi_{k}}\bra{\psi_{k}}\right)  +  \sum_{k=1}^{M'} (1-t) q'_k   \mudown \left( \ket{\psi'_{k}}\bra{\psi'_{k}}\right) \in \Ud(\rho)$. Hence,  $\Ud(\rho)$ is a convex set.

Now, we consider the set $\Uo(\rho)$. 
Let $u, u' \in \Uo(\rho)$. Again, given $t \in \left(0, 1 \right)$, we consider the ordered probability vector $u_t = t u  + (1-t) u'$. Also, we consider a pure state $\ket{u_t}\bra{u_t}$, such that  
$\mudown\left( \ket{u_t}\bra{u_t}\right) = u_t $.

From Lemma \ref{lemma:cota}, we know that there are two probability vectors $v, v' \in\Ud(\rho)$, such that $u\preceq v$ and $u'\preceq v'$. If we define the ordered probability vector $v_t =t v  + (1-t) v'$, then, from Lemma~\ref{lemma:resultado_auxiliar}, we have $u_t = t u  + (1-t) u' \preceq  t v  + (1-t) v'= v_t$.
Since $\Ud(\rho)$ is convex, $v_t \in \Ud(\rho)$. By definition of the set $\Ud(\rho)$, there is a pure state decomposition $\{q_k, \ket{\phi_k}\}_{k = 1}^M$
of $\rho$, such that $v_t = \sum_{k = 1}^{M}  q_k \mu^\downarrow(\ket{\phi_k}\bra{\phi_k})$. 

Summing up, given the pure state $\ket{u_t}\bra{u_t}$, we have that
$u_t = \mudown\left( \ket{u_t}\bra{u_t}\right) \preceq v_t = \sum_{k = 1}^{M}  q_k \mu^\downarrow(\ket{\phi_k}\bra{\phi_k})$, with $\rho = \sum_{k = 1}^{M}  q_k \ket{\phi_k}\bra{\phi_k}$.
Finally, from Prop. \ref{lemma:theorem4du2019}, we conclude $\ket{u_t}\bra{u_t} \io \rho$, which implies $u_t = t u  + (1-t) u' \in \Uo(\rho)$. Hence,  $\Uo(\rho)$ is a convex set.

\end{proof}

\begin{repproposition}{prop:supremum_equality}
	$\bigvee\Ud(\rho) = \bigvee\Uo(\rho)$.
\end{repproposition}

\begin{proof}
	From Lemma~\ref{lemma:inclusion}, we have $\Ud(\rho)\subseteq\Uo(\rho)$. Then, $ \bigvee\Ud(\rho) \preceq \bigvee\Uo(\rho)$.
	In addition, from Lemma~\ref{lemma:cota}, we have  $ \bigvee\Uo(\rho) \preceq \bigvee\Ud(\rho)$. 
	Therefore, since the majorization relation is antisymmetric, we obtain $\bigvee\Ud(\rho) = \bigvee\Uo(\rho)$.
\end{proof}

\begin{repproposition}{lemma:maximum_sets}
	$\vc(\rho) \in \Ud(\rho) \iff \vc(\rho) \in \Uo(\rho)$.
\end{repproposition}

\begin{proof}
	~ 

	\begin{itemize}
		\item[$\left( \Longrightarrow\right) $] Suppose $\vc(\rho) \in \Ud(\rho)$. Then, from Lemma~\ref{lemma:inclusion}, it follows that  $\vc(\rho) \in \Uo(\rho)$.
		
		\item[$\left( \Longleftarrow\right) $] Suppose	$\vc(\rho) \in \Uo(\rho)$. From Lemma~\ref{lemma:cota}, exists $u' \in \Ud(\rho)$ such that $\vc(\rho)  \preceq u'$. From Prop.\ref{prop:supremum_equality}, we also have that $u' \preceq \vc(\rho)$. Then, $u ' = \vc(\rho)   \in  \Ud(\rho)$.
		
	\end{itemize}
	
\end{proof}

\begin{repproposition}{prop:incoheren_state}
	$\rho$  is incoherent $\iff \vc(\rho)=(1,0,\ldots,0)$.
\end{repproposition}

\begin{proof}

	~ 

	\begin{itemize}
		\item[$\left( \Longrightarrow\right) $] 	
		Let $\rho \in \I$ be an incoherent state. By definition, $\rho$ is diagonal in the incoherent basis, that is, $\rho= \sum_{i=0}^{d-1} p_i \ket{i}\bra{i}$.
		Since, 	$\{p_i, \ket{i}\} \in \D(\rho)$ and  $\sum_i p_i \mu^\downarrow(\ket{i}\bra{i}) = (1,0,\ldots,0)\in \Ud(\rho)$, then 	$\vc(\rho)=(1,0,\ldots,0)$.

		\item[$\left( \Longleftarrow\right) $] 
		
		Let $\rho \in \S(\H)$ be such that $\vc(\rho)=(1,0,\ldots,0)$. 
		To prove the converse statement, we appeal to \textit{reductio ad absurdum} by assuming that $\rho$ is a coherent state.	
		From Prop. \ref{prop:supremum_equality}, we have that $\vc(\rho) = \bigvee \Uo(\rho)$.
		
		According to the formula of the supremum~\cite{Bosyk2019}, the first entry of $\vc(\rho)$ is given by the supremum of the first entries of the vectors of $\Uo(\rho)$, i.e., 
		\begin{equation}
		\label{eq:supremum_1_def}
		\left( \vc(\rho)\right) _1=  \bigvee \left\{ \left(  \mu^\downarrow(\ket{\psi}\bra{\psi})\right)_1  : \ket{\psi}\bra{\psi} \in \O(\rho)  \right\},
		\end{equation}
		where
		\begin{equation}
		\label{eq:coherence_vector_ensamble}
		\left(  \mu^\downarrow(\ket{\psi}\bra{\psi}) \right) _1=   \max_{0 \leq i \leq d-1} \left|\braket{i| \psi}\right|^2.
		\end{equation}
		Then, 
		\begin{equation}
		\left( \vc(\rho)\right) _1=  \max_{0 \leq i \leq d-1} \bigvee \left\{  \left|\braket{i| \psi}\right|^2   : \ket{\psi}\bra{\psi} \in \O(\rho)  \right\} .
		\end{equation}
		For each $0 \leq i \leq d-1$, we consider the function  $f_i: \O(\rho) \to \R$, given by $f_i( \ket{\psi}\bra{\psi})= \left|\braket{i| \psi}\right|^2 $. Since $\O(\rho)$ is compact (see Lemma~\ref{lemma:compact}) and $f_i$ is continuous, 
		there exists a pure state $\ket{\psi_i}\bra{\psi_i} \in \O(\rho)$ which is the maximum of $f_i$ in $O(\rho)$, i.e.,
		\begin{equation}
		f_i( \ket{\psi_i}\bra{\psi_i})= \max \left\{  \left|\braket{i| \psi}\right|^2   : \ket{\psi}\bra{\psi} \in \O(\rho)  \right\} .
		\end{equation}
		Therefore, if we define  $f_{i^*}( \ket{\psi_{i^*}}\bra{\psi_{i^*}})  = \max_{0 \leq i \leq d-1} f_i( \ket{\psi_i}\bra{\psi_i})$, we have
		\begin{equation}
		\label{eq:supremum_1_max}
		\left( \vc(\rho)\right) _1 = f_{i^*}( \ket{\psi_{i^*}}\bra{\psi_{i^*}}),
		\end{equation}
		with $\ket{\psi_{i^*}}\bra{\psi_{i^*}} \in  \O(\rho)$. By hypothesis, $\vc(\rho)=(1,0,\ldots,0)$, then 
		$f_{i^*}(\ket{\psi_{i^*}}\bra{\psi_{i^*}}) = \left|\braket{i^*| \psi_{i^*}}\right|^2  = 1$.
		This implies that $\ket{{\psi_{i^*}}}\bra{{\psi_{i^*}}} \in \I$. 
		
		Summing up, $\ket{\psi_{i^*}}\bra{\psi_{i^*}}$ is an incoherent pure state that can be transformed into the coherent state $\rho$ by means of an  incoherent operation, but this is absurd. Therefore, $\rho$ has to be incoherent. 
				
	\end{itemize}
		
\end{proof}

\begin{repproposition}{prop:coherence_vector_maximum_coherence}
	$\rho$ is maximally coherent $\iff \vc(\rho) = \left(\frac{1}{d}, \ldots, \frac{1}{d}\right)$.
\end{repproposition}

\begin{proof}

	~ 

	\begin{itemize}
		\item[$\left( \Longrightarrow\right) $] 	
			
		Let $\rho \in \S(\H)$ be an arbitrary maximally coherent state, that is,  $\rho = U_{\mathrm{IO}} \ket{\Psi^{\mcs}}\bra{\Psi^{\mcs}} U^\dag_{\mathrm{IO}}$, with $\ket{\Psi^{\mcs}} = \sum_{i=0}^{d-1} \frac{1}{\sqrt{d}} \ket{i}$ and $U_{\mathrm{IO}} = \sum_{i= 0}^{d-1} e^{\imath \theta_{i}} \ket{\pi(i)}\bra{i}$, where $\theta_{i} \in \R$ and $\pi$ is a permutation acting on the set $\{0,\ldots,d-1\}$.
		Since $|\bra{i} U_{\mathrm{IO}} \ket{\Psi^{\mcs}}|^2 = 1/d$ for all $i \in \{0,\ldots,d-1\}$, we have $\vc(\rho) = \mu\left(U_{\mathrm{IO}} \ket{\Psi^{\mcs}}\bra{\Psi^{\mcs}} U^\dag_{\mathrm{IO}}\right) = (1/d, \ldots, 1/d)$.

		\item[$\left( \Longleftarrow\right) $] 
		
		Let $\rho \in \S(\H)$ be such that $\vc(\rho) = \left(\frac{1}{d}, \ldots, \frac{1}{d}\right)$.

		Firstly, we consider the pure state case, i.e., $\rho = \ket{\psi}\bra{\psi}$. 
		The coherence vector of $\rho$ is given by $\vc(\rho) = \mudown(\ket{\psi}\bra{\psi}) =\left(\frac{1}{d}, \ldots, \frac{1}{d}\right)$. From Def.~\ref{def:coherence_vector_pure}, it follows $|\braket{i|\psi}|^2 = 1/d$ for all $i \in \{0,\ldots,d-1\}$.
		Therefore, $\ket{\psi} = U_{\mathrm{IO}}\ket{\Psi^{\mcs}}$, with $U_{\mathrm{IO}} = \sum_{i= 0}^{d-1} e^{\imath \theta_{i}} \ket{i}\bra{i}$  and $\theta_{i} \in \R$. This implies that $\rho$ is a maximally coherent state.
		
		Secondly, we are going to show that $\rho$ has to be a pure state. We appeal to \textit{reductio ad absurdum} by assuming that $\rho$ is a mixed state.	
		Let $\{q_k, \ket{\psi_k}\}_{k = 1}^M$ be an arbitrary pure state decomposition of $\rho$, i.e., $\rho = \sum_{k= 1}^M q_k \ket{\psi_k}\bra{\psi_k}$. 	On the one hand, by definition of $\vc(\rho)$, we have $\sum_{k = 1}^M q_k \mudown(\ket{\psi_k}\bra{\psi_k}) \preceq (1/d, \ldots, 1/d)$.  
		On the other hand, since $(1/d, \ldots, 1/d)$ is the bottom of the majorization lattice, we have $(1/d, \ldots, 1/d) \preceq \sum_{k = 1}^M q_k \mudown(\ket{\psi_k}\bra{\psi_k})$.
		Then, $\sum_{k = 1}^M q_k \mudown(\ket{\psi_k}\bra{\psi_k}) = (1/d, \ldots, 1/d)$.
		Moreover, the probability vector $(1/d, \ldots, 1/d)$ is an extreme point of the $d-1$-simplex, which implies that  $\mudown(\ket{\psi_k}\bra{\psi_k}) = (1/d, \ldots, 1/d)$ for all $k \in \{1, \ldots,  M\}$. 
		Then, states $\ket{\psi_k}\bra{\psi_k}$ have to be maximally coherent states.
		Therefore, we conclude that any pure state decomposition of $\rho$ has to be formed by maximally coherent pure states.  
	
		In particular, we consider the spectral decomposition of $\rho$,
		\begin{equation}
		\rho = \sum_{j=1}^{d} \lambda_j \ket{e_j}\bra{e_j}.
		\end{equation}
		The eigenvectors have to be maximally coherent pure states. 
		Since, by hypothesis $\rho$ is a mixed state, there are at least two eigenvalues different from zero. Without loss of generality, we consider $\lambda_1, \lambda_2 >0$. In terms of the incoherent basis, we have $\ket{e_1} = \sum_{i=0}^{d-1} \frac{e^{\imath \alpha_i}}{\sqrt{d}} \ket{i}$ and $\ket{e_2} = \sum_{i=0}^{d-1} \frac{e^{\imath \beta_i}}{\sqrt{d}} \ket{i}$, with $\alpha_i, \beta_i \in \R$, for all $i \in \{0,\ldots,d-1\}$.
		
		According to the Schr\"odinger mixture theorem (see \eg \cite{Hughston1993,Nielsen2000}), any ensemble $\{p_k, \ket{\phi_k}\}_{k = 1}^M$ is a pure state decomposition of $\rho$ if, and only if, there exists a unitary matrix $U$ such that
		\begin{equation}
		\label{eq:SchTh}
		\ket{\phi_k} = \frac{1}{\sqrt{p_k}} \sum_{j=1}^{d}  \sqrt{\lambda_j} U_{k,j} \ket{e_j}.
		\end{equation}

		We consider a $d \times d$ unitary matrix of the form
		\begin{equation}
		U = \begin{pmatrix}
		\begin{matrix}
		U_{11} & U_{1,2} \\
		-U^*_{12} & U_{1,1}
		\end{matrix}
		& \rvline & \bigzero \\
		\hline
		\bigzero & \rvline &
		I_{d-2}
		\end{pmatrix}
		\end{equation}
		with $U_{1,1} = \sqrt{\frac{\lambda_2}{\lambda_1 + \lambda_2}}$ and $ U_{1,2} =  - e^{\imath (\alpha_0-\beta_0)} \sqrt{\frac{\lambda_1}{\lambda_1+\lambda_2}} $.
		Then, the first state takes the form
		\begin{equation}
		\ket{\phi_1} = \frac{1}{\sqrt{p_1}} \left(  \sqrt{\lambda_1} U_{1,1} \ket{e_1} +\sqrt{\lambda_2} U_{1,2} \ket{e_2}\right), 
		\end{equation}
		and, taking into account the expression of $\ket{e_1}$ and $\ket{e_2}$ in the incoherent basis, we obtain 
		\begin{equation}
		\braket{0|\phi_1} = \frac{1}{\sqrt{d p_1}}  \left(e^{\imath \alpha_0}  \sqrt{\lambda_1} U_{1,1}  +e^{\imath \beta_0}\sqrt{\lambda_2} U_{1,2} \right) = 0, 
		\end{equation}
		which is in contradiction with $\ket{\phi_1}$ being a maximally coherent state. Therefore, $\rho$ cannot be a mixed state, it has to be a pure state.

	\end{itemize}

\end{proof}

\begin{repproposition}{prop:convexity_mu}
	Let $\rho = \sum_{k = 1}^{M}  p_k \ket{\psi_k} \bra{\psi_k}$. Then,
	\begin{equation}
	\sum_{k = 1}^{M}  p_k \vc\left(\ket{\psi_k} \bra{\psi_k}\right) \preceq \vc(\rho). 
	\end{equation}
\end{repproposition}

\begin{proof}
	
Let $\rho = \sum_{k = 1}^{M}  p_k \ket{\psi_k} \bra{\psi_k}$.
We have $\sum_{k = 1}^{M}  p_k \mudown\left(\ket{\psi_k} \bra{\psi_k}\right) = \sum_{k = 1}^{M}  p_k \vc\left(\ket{\psi_k} \bra{\psi_k}\right) \in \Ud(\rho)$.
Then, by definition of the supremum, $\sum_{k = 1}^{M}  p_k \vc\left(\ket{\psi_k} \bra{\psi_k}\right) \prec \vc(\rho)$.

\end{proof}

\begin{repproposition}{prop:monotonocity_vectorcoherence}
	Let $\rho$ and $\sigma$ be two arbitrary quantum states. Then,
	\begin{multline}
	\rho \io \sigma  \!\implies \! \forall \{q_k,\ket{\psi_k}\}_{k = 1}^M \!\in\! \D(\rho), \exists \{r_l,\ket{\phi_l}\}_{l \in L} \!\in\! \D(\sigma): \\
	\sum_{k = 1}^{M}  q_k \mu^\downarrow(\ket{\psi_k}\bra{\psi_k}) \preceq  \sum_{l \in L}  r_{l} \mu^\downarrow(\ket{\phi_{l}}\bra{\phi_{l}}). 
	\end{multline}
\end{repproposition}

\begin{proof}
	
Let $\Lambda$ be an incoherent operation, with incoherent Kraus operators $\{K_n\}_{n = 1}^N$, such that $\sigma= \Lambda(\rho) = \sum_{n = 1}^{N}  K_n \rho K^\dag_n$.
Let $\{q_k, \ket{\psi_k} \}_{k = 1}^M$ be an arbitrary pure state decomposition of $\rho$, that is, $\rho = \sum_{k = 1}^{M}  q_k \ket{\psi_k}\bra{\psi_k}$.
Then, we have $\sigma =\Lambda(\rho) =  \sum_{n = 1}^{N} \sum_{k= 1}^M q_k p_{n,k} \ket{\phi_{n,k}}\bra{\phi_{n,k}}$, with $p_{n,k} = \Tr(K_n \ket{\psi_k}\bra{\psi_k} K_n^\dag)$ and $\ket{\phi_{n,k}} = K_n \ket{\psi_k}/\sqrt{p_{n,k}}$.
	
In particular, for each $\ket{\psi_k}\bra{\psi_k}$, we have $\ket{\psi_k}\bra{\psi_k} \io \sum_{n= 1}^{N}  p_{n,k} \ket{\phi_{n,k}}\bra{\phi_{n,k}}$,
Then, according to Eq.~\eqref{eq:lemma4zhu} (\cite[Lemma 4]{Zhu2017}),  
\begin{equation}
\mu^\downarrow(\ket{\psi_k}\bra{\psi_k}) \preceq  \sum_{n = 1}^{N}  p_{n,k} \mu^\downarrow(\ket{\phi_{n,k}}\bra{\phi_{n,k}}). 
\end{equation}
Applying Lemma~\ref{lemma:resultado_auxiliar} for the sequences of ordered probability vectors $\{ \mu^\downarrow(\ket{\psi_k}\bra{\psi_k}) \}_{k = 1}^M$ and $ \left\lbrace \sum_{n} p_{n,k} \mu^\downarrow(\ket{\phi_{n,k}}\bra{\phi_{n,k}})\right\rbrace_{k = 1}^M$,
we obtain
\begin{equation}
\label{eq:lemma4zhu_mixed}
\sum_{k = 1}^{M}  q_k \mu^\downarrow(\ket{\psi_k}\bra{\psi_k}) \preceq  \sum_{n = 1}^{N} \sum_{k= 1}^M q_k p_{n,k} \mu^\downarrow(\ket{\phi_{n,k}}\bra{\phi_{n,k}}), 
\end{equation}
where $q_k \geq 0$ and $\sum_{k = 1}^{M}  q_k =1$.
Defining
$r_l= q_k p_{n,k}$, $\ket{\phi_{l}} = \ket{\phi_{n,k}}$ and $L = \{(n,k) : 1 \leq n \leq N,  1 \leq k \leq M \}$, we can rewrite expression \eqref{eq:lemma4zhu_mixed} as 
\begin{equation}
\label{eq:lemmazhu_mixed_bis}
\sum_{k = 1}^{M}  q_k \mu^\downarrow(\ket{\psi_k}\bra{\psi_k}) \preceq  \sum_{l \in L} r_l \mu^\downarrow(\ket{\phi_{l}}\bra{\phi_{l}}), 
\end{equation}
with $\{r_l,\ket{\phi_l}\}_{l \in L} \in \D(\sigma)$.
Since the  majorization relation \eqref{eq:lemmazhu_mixed_bis} is valid for any pure sate decomposition  of $\rho$, we conclude that for each  
$\{q_k,\ket{\psi_k}\}_{k = 1}^M \in \D(\rho)$,
there exists a pure state decomposition  $\{r_l,\ket{\phi_l}\}_{l \in L} \in \D(\sigma)$, such that relation  \eqref{eq:lemmazhu_mixed_bis} is satisfied.
	
\end{proof}

\begin{repcorollary}{cor:monotonocity_vectorcoherence 2}
	Let $\rho$ and $\sigma$ be two arbitrary quantum states. Then, 	
	\begin{equation}
		\rho \io \sigma  \implies \vc(\rho) \preceq \sum_{n = 1}^{N}  p_n \vc(\sigma_n),
	\end{equation}
	with $p_n = \Tr \left( K_n \rho K^\dag_n \right) $ and $\sigma_n = K_n \rho K^\dag_n/p_n$, where $\{K_n\}_{n = 1}^N$ are incoherent Kraus operators such that $\sigma= \sum_{n = 1}^{N}  K_n \rho K^\dag_n$.
\end{repcorollary}

	\begin{proof}	
	Let $\Lambda$ be an incoherent operation, with incoherent Kraus operators $\{K_n\}_{n=1}^N$, such that $\sigma= \Lambda(\rho) = \sum_{n = 1}^{N}  K_n \rho K^\dag_n$, and define $p_n = \Tr\left(K_n \rho K^\dag_n\right)$ and $\sigma_n = K_n \rho K^\dag_n/p_n$.
	
	For any arbitrary pure state decomposition  $\{q_k, \ket{\psi_k} \}_{k = 1}^M$ of $\rho$, we can write $\sigma_n = \sum_{k = 1}^{M}  q_k p_{n,k} \ket{\phi_{n,k}}\bra{\phi_{n,k}} / p_n$, with $p_n = \sum_k q_k p_{n,k}$.
	Since $\sum_{k = 1}^{M}  q_kp_{n,k} \mudown(\ket{\phi_{n,k}}\bra{\phi_{n,k}})/p_n \in \Ud(\sigma_n)$, then	%
	\begin{equation}
	\sum_{k = 1}^{M}  \frac{q_kp_{n,k}}{p_n} \mudown(\ket{\phi_{n,k}}\bra{\phi_{n,k}}) \preceq \vc(\sigma_n).
	\end{equation} 
	Multiplying by $p_n$,  summing over $n$, and using ~\eqref{eq:lemma4zhu_mixed}, we obtain 
	\begin{align}
	\sum_{k = 1}^{M}  q_k \mudown(\ket{\psi_{k}}\bra{\psi_{k}}) &\preceq \sum_{n = 1}^{N} \sum_{k= 1}^M q_kp_{n,k} \mudown(\ket{\phi_{n,k}}\bra{\phi_{n,k}}) \\
	&\preceq \sum_{n = 1}^{N}  p_n \vc(\sigma_n).
	\end{align} 
	The last majorization relation does not depend on the pure state decomposition of $\rho$, then $\sum_{n = 1}^{N}  p_n \vc(\sigma_n)$ is also an upper bound of $\Ud(\rho)$. 
	Therefore, by definition of supremum, we conclude that $\vc(\rho) \preceq \sum_n p_n \vc(\sigma_n)$.

\end{proof}

\begin{repcorollary}{cor:monotonocity_vectorcoherence_supremum}
	Let $\rho$ and $\sigma$ be two arbitrary quantum states. Then,
	\begin{equation}
	\rho \io \sigma  \implies \vc(\rho) \preceq \vc(\sigma).
	\end{equation}
\end{repcorollary}

\begin{proof}
	Since $\rho \io \sigma$, from
	Prop.~\ref{prop:monotonocity_vectorcoherence}, we have that, for all $ \{q_k,\ket{\psi_k}\}_{k = 1}^M \in \D(\rho)$, there is a $\{r_l,\ket{\phi_l}\}_{l \in L} \in \D(\sigma)$, such that
	\begin{equation}
	\sum_{k = 1}^{M}  q_k \mu^\downarrow(\ket{\psi_k}\bra{\psi_k}) \preceq  \sum_{l \in L}  r_{l} \mu^\downarrow(\ket{\phi_{l}}\bra{\phi_{l}}).
	\end{equation}
	Then, from the definition of the supremum we have 
	\begin{equation}
	\sum_{k = 1}^{M}  q_k \mu^\downarrow(\ket{\psi_k}\bra{\psi_k}) \preceq \vc(\sigma).
	\end{equation}
	This implies that $\vc(\sigma)$ is an upper bound of the set $\Ud(\rho)$. 
Therefore, by definition of  $\vc(\rho)$, we have $\vc(\rho) \preceq \vc(\sigma)$.

\end{proof}

\begin{repproposition}{prop:cohrence_measures_pure_conv}	
	For any function $f \in \F$, the coherence vector measure $C^{\cv}_{f}$ satisfies conditions \ref{c1:nonneg}--\ref{c5:normalization}.
\end{repproposition}

\begin{proof}
	\hfill
	\begin{enumerate}[label=(\subscript{\rm{C}}{{\arabic*}})]
		\item By Prop.~\ref{prop:incoheren_state}, if $\rho \in \I$, then $\vc(\rho) = (1,0,\ldots,0)$. Therefore,  $C^{\cv}_{f}(\rho)=f(1,0,\ldots,0) =0$. 

		\item 
		Since $\rho \io	\Lambda(\rho)$, from Cor.
		\ref{cor:monotonocity_vectorcoherence_supremum}, we obtain 	 $\vc(\rho ) \preceq \vc(\Lambda(\rho))$.
		Moreover, $f$ is symmetric and concave, then $f$ is also Schur-concave, which implies that 
		$f(\vc(\rho )) \geq f(  \vc(\Lambda(\rho)))$.  Finally, we conclude that
		$C^{\cv}_f(\rho)  \geq C^{\cv}_f(\Lambda(\rho))$. 

	\item 
	Let $\rho \in \S(\H)$ be an arbitrary quantum state and $\Lambda$ an incoherent operation, with incoherent Kraus operators $\{K_n\}_{n= 1}^N$, $p_n = \Tr\left(K_n \rho K^\dag_n\right)$ and $\sigma_n = K_n \rho K^\dag_n/ p_n$.
	If we define $\sigma = \Lambda(\rho)$, from Cor.
	\ref{cor:monotonocity_vectorcoherence 2}, Eq.  \eqref{eq:strong_monotonocity_vectorcoherence_supremum}, we obtain $\vc(\rho) \preceq \sum_{n= 1}^N p_n \vc(\sigma_n)$. 
	Then, we have
	\begin{equation*}
	 \sum_{n= 1}^N p_n f \left( \vc(\sigma_{n} )\right)  \leq f\left( \sum_{n= 1}^N p_n \vc(\sigma_n)\right)  \leq f(\vc(\rho)),
	\end{equation*} 
	where in the first inequality we have used the concavity of $f$ and in the second one the
	Schur-concavity.
	Finally, taking into account the coherence vector definition, we conclude
	\begin{equation*}
	 \sum_{n= 1}^N p_n C^{\cv}_f(\sigma_n)  \leq  C^{\cv}_f(\rho).
	\end{equation*}

	\item 
	
	Let $\rho$ a maximally coherent state and $\sigma$ an arbitrary state. Due to Prop.~\ref{prop:coherence_vector_maximum_coherence}, $\vc(\rho) = (1/d,\ldots,1/d)$. 
	Moreover, since $ \arg\max_{u \in \mathbb{R}^d} f(u)= (1/d,  \ldots, 1/d) $, we have  $f(\vc(\rho)) =  f(1/d,\ldots,1/d) \geq  f(\nu(\sigma))$. This implies that $C_f(\rho) \geq C_f(\sigma)$. 
	Therefore, we conclude that $\arg\max_{\rho \in S(\H)} C_f(\rho)$ is reached at maximally coherent states.

\end{enumerate}	
\end{proof}

\begin{repproposition}{prop:relation_coherence_measures_op_sup}
	The following statements are equivalent: 
	\begin{enumerate}
		\item  There exists an optimal pure state decomposition of $\rho$, i.e., $\vc(\rho) \in \Ud(\rho)$. \label{st1}
		\item $C^{\cv}_f(\rho) =  C^{\op}_f(\rho)$ for all $f \in \F$. \label{st2} 
		\item $C^{\cv}_f(\rho) =  C^{\op}_f(\rho)$ for some $f \in \F$ strictly Schur-concave. \label{st3}
	\end{enumerate}
\end{repproposition}

\begin{proof}

~

\begin{itemize}[itemindent=2em]
\item[$\left(\ref{st1} \implies \ref{st2}\right)$] 	
\hfill

Let $f \in \F$. 
On the one hand, by Prop.~\ref{prop:relation_top}, we have $C^{\op}_f(\rho) \geq C^{\cv}_f(\rho)$.
On the other hand, if there exists an optimal pure state decomposition of $\rho$, then $\vc(\rho) \in \Ud(\rho)$. From Lemma \ref{lemma:inclusion}, we have  $\vc(\rho) \in \Uo(\rho)$.
By definition of the top measure, $C^{\op}_f(\rho) \leq f(u)$ for all $u \in \Uo(\rho)$.
In particular, $C^{\op}_f(\rho) \leq f(\nu(\rho)) = C^{\cv}_f(\rho)$. 
Finally, we conclude $C^{\cv}_f(\rho) = C^{\op}_f(\rho)$, which is valid for all $f \in \F$.

\item[$\left(\ref{st2} \implies \ref{st3}\right)$]
\hfill

Trivial.

\item[$\left(\ref{st3} \implies \ref{st1}\right)$]	
\hfill

Let $f \in \F$ be a strictly Schur-concave function such that  
$C^{\cv}_f(\rho) =  C^{\op}_f(\rho)$.

Notice that $C^{\op}_{f}$ can be written as 
\begin{equation}
C^{\op}_{f}(\rho) = \min_{ u \in \Uo(\rho) } f(u).
\end{equation}
We denote the probability vector where the minimum is reached as $\tilde{u}$. Then, $f(\vc(\rho))= C^{\cv}_f(\rho) =  C^{\op}_f(\rho) = f(\tilde{u})$. Since $f$
is strictly Schur-concave, then by Lemma \ref{lemma:strict_schur}, we have $\vc(\rho) = \tilde{u}\in \Uo(\rho)$.
Finally, by Lemma \ref{lemma:maximum_sets}, $\vc(\rho) \in \Ud(\rho)$, i.e., there exists an optimal pure state decomposition of $\rho$.

\end{itemize}
\end{proof}

\begin{repproposition}{prop:relation_coherence_measures_sup_croof}
	If there exists an optimal pure state decomposition of $\rho$, then $C^{\cv}_f(\rho) \geq C^{\croof}_f(\rho)$.
\end{repproposition}	
	
\begin{proof}
Let $\{\tilde{q}_k,\ket{\tilde{\psi}_k}\}_{k = 1}^M \in \D(\rho)$ be an optimal pure sate decomposition of $\rho$. Thus, $\vc(\rho)= \sum_{k= 1}^M \tilde{q}_k \mu^\downarrow(\ket{\tilde{\psi}_k}\bra{\tilde{\psi}_k})$.
Let $f \in \F$, then 
\begin{align}
	 C_f^{\cv}(\rho)&= f(\nu(\rho)) = f\left(\sum_{k= 1}^M \tilde{q}_k \mu^\downarrow(\ket{\tilde{\psi}_k}\bra{\tilde{\psi}_k})\right) \\
	 &\geq \sum_{k= 1}^M \tilde{q}_k f\left(  \mu(\ket{\tilde{\psi}_k}\bra{\tilde{\psi}_k})\right) \\
	 &\geq   \inf_{\left\lbrace  q_k, \ket{ \psi_k} \right\rbrace_{k = 1}^M  \in \D(\rho) } \sum_{k= 1}^M q_k 	 f(\mu(\ket{\psi_k}\bra{\psi_k})) \\
	 &= C^{\croof}_{f}(\rho),
\end{align}
where the first inequality comes from the concavity and symmetric properties of $f$, and the second one comes from the definition of the convex roof measure. 
\end{proof}

\begin{repproposition}{prop:linear f}
	Let $f\in\F$ be such that $f|_{\Delta_d^\downarrow}=c+\ell$, where $c\in\mathbb{R}$ and $\ell : \Delta_d^\downarrow \to \mathbb{R}$
	is a linear function. Then, $C^{\cv}_f \leq C^{\croof}_f=C^{\op}_f$.
\end{repproposition}

\begin{proof}
Let $\rho \in \S(\H)$.
On the one hand, by definition of the convex roof measure, we have
\begin{align}
C^{\croof}_{f}(\rho) &= \inf_{\left\lbrace  q_k, \ket{ \psi_k} \right\rbrace_{k = 1}^M  \in \D(\rho) } \sum_{k= 1}^M q_k 	 f(\mu(\ket{\psi_k}\bra{\psi_k})) \\
&= \sum_{k= 1}^M \tilde{q}_k f\left(  \mudown(\ket{\tilde{\psi}_{k}}\bra{\tilde{\psi}_{k}}) \right), 
\end{align}
with $\left\lbrace \tilde{q}_k, \ket{\tilde{ \psi}_k} \right\rbrace_{k = 1}^M  \in \D(\rho)$ the pure state decomposition of $\rho$ where the minimum is reached.
Taking into account the form of $f$, we get 
\begin{align}
\label{eq:cr_top_igualdad1}
C^{\croof}_{f}(\rho) &=\sum_{k= 1}^M \tilde{q}_k \left( c +  \ell \left( \mudown(\ket{\tilde{\psi}_{k}}\bra{\tilde{\psi}_{k}}) \right)  \right) \\
&=  c + \ell \left( \sum_{k= 1}^M \tilde{q}_k \mudown(\ket{\tilde{\psi}_{k}}\bra{\tilde{\psi}_{k}}) \right) \\
&= f\left( \sum_{k= 1}^M \tilde{q}_k \mudown(\ket{\tilde{\psi}_{k}}\bra{\tilde{\psi}_{k}}) \right),
\end{align}
where we have used the linearity of $\ell$ and the condition $\sum_{k = 1}^M q_k= 1$.
	
On the other hand, by definition of $\nu(\rho)$ and Schur-concavity of $f$, we have 
\begin{equation}
C^{\cv}_f(\rho)= f(\nu(\rho))\leq f\left( \sum_{k= 1}^M q_k \mudown(\ket{\psi_{k}}\bra{\psi_{k}}) \right), 
 \forall \{q_k,\ket{\psi}_k\}_{k = 1}^M \in \D(\rho).
\end{equation}
In particular, 
\begin{equation}
C^{\cv}_f(\rho) \leq f\left( \sum_{k= 1}^M \tilde{q}_k \mudown(\ket{\tilde{\psi}_{k}}\bra{\tilde{\psi}_{k}}) \right).
\end{equation}
Therefore,  we conclude $C^{\cv}_f(\rho)\le C^{\croof}_f(\rho)$.

In order to prove the equality part of the proposition, first we note that $C^{\croof}_f(\rho) \leq C^{\op}_f(\rho)$, see Ineq. \eqref{eq:relation_top}.
Moreover, by definition of the top measure, we have 
\begin{equation}
C^{\op}_{f}(\rho)  \leq f(\mudown(\ket{\psi}\bra{\psi}), ~~ \forall \ket{\psi}\bra{\psi} \in \O(\rho) .
\end{equation}
Since $\Ud(\rho)\subseteq\Uo(\rho)$ (see Lemma \ref{lemma:inclusion}), 
we have that  
$\sum_{k= 1}^M \tilde{q}_k \mudown(\ket{\tilde{\psi}_{k}}\bra{\tilde{\psi}_{k}}) \in  \Uo(\rho)$. 
Then, there is $\ket{\tilde{\psi}}\bra{\tilde{\psi}} \in \O(\rho)$, such that $\mudown(\ket{\tilde{\psi}}\bra{\tilde{\psi}})=  \sum_{k= 1}^M \tilde{q}_k \mudown(\ket{\tilde{\psi}_{k}}\bra{\tilde{\psi}_{k}})$. Therefore, 
\begin{equation}
\label{eq:cr_top_igualdad2}
C^{\op}_{f}(\rho) \leq f(\mudown(\ket{\tilde{\psi}}\bra{\tilde{\psi}}) =  f\left( \sum_{k= 1}^M \tilde{q}_k \mudown(\ket{\tilde{\psi}_{k}}\bra{\tilde{\psi}_{k}}) \right) .
\end{equation}
Then, from \eqref{eq:cr_top_igualdad1} and \eqref{eq:cr_top_igualdad2}, we have  $C^{\op}_f(\rho) \leq C^{\croof}_f(\rho)$.  Finally, we conclude that $C^{\op}_f(\rho) = C^{\croof}_f(\rho)$.
\end{proof}

	\begin{repproposition}{prop:exsitence_optimalpsd_roof}
		If there exists an optimal pure state decomposition of $\rho$, i.e., $\vc(\rho) \in \Ud(\rho)$, then
		\begin{equation}
			C^{\cv}_{f_k}(\rho) =  C^{\croof}_{f_k}(\rho),  
		\end{equation}
		with $f_k(u) = 1-\sum^{k-1}_{i=0} u^\downarrow_i \in \F$, for all $k \in \{1, \ldots, d-1\}$.	
	\end{repproposition}
	
	\begin{proof}
		
	On the one hand, let us recall that if $\vc(\rho) \in \Ud(\rho)$, then $S_k(\rho) =s_k(\vc(\rho))$ for all $k \in \{1, \ldots, d-1\}$, where  $S_k(\rho) =\sup_{u\in\Ud(\rho)} s_k(u)$ and $s_j(u) = \sum^{j-1}_{i=0} u_i$, with $u = \left(u_0, \ldots, u_{d-1}\right)$.
	
	On the other hand, we note that $f_k(u) =  1- s_k(u^\downarrow)$. 
	Therefore, 
	\begin{align}
		\label{eq:roof_cv_fk}
		C^{\croof}_{f_k}(\rho) &=1-S_k(\rho) \ \text{and} \\ 
		C_{f_k}^{\cv}(\rho) & =1-s_k(\nu(\rho)).
	\end{align}
	
	Summarizing, $\vc(\rho) \in \Ud(\rho)$ implies $C^{\croof}_{f_k}(\rho) = C_{f_k}^{\cv}(\rho)$ for all $k \in \{1, \ldots, d-1\}$.
	
	\end{proof}


\end{document}